\documentclass[12pt]{article}
\usepackage{natbib}
\usepackage{undertilde,graphicx,bm,multirow}
\usepackage{amsthm,amsfonts,amssymb,color,comment,pdfpages,csvsimple,colortbl}
\usepackage[tbtags]{amsmath}
\usepackage{comment}
\usepackage[hidelinks]{hyperref} 
\usepackage[top=1in,bottom=1in,left=1in,right=1in]{geometry}
\usepackage{setspace}
\usepackage{algorithm}
\usepackage{algpseudocode}
\usepackage{enumerate}
\usepackage{subcaption}
\newtheorem{thm}{Theorem}
\newtheorem{lem}{Lemma}

\newtheorem*{fact}{Fact}
\newtheorem{Def}{Definition}

\numberwithin{equation}{section}

\definecolor{Red}{rgb}{.9,0,0}

\begin{document}
\title{Bayesian inference for Gaussian graphical models beyond decomposable graphs}
\author{ Kshitij Khare, \emph{University of Florida, USA}\\
  Bala Rajaratnam, \emph{Stanford University, USA}\\
  Abhishek Saha, \emph{University of Florida, USA}\\  
  } \date{}

\maketitle

\onehalfspacing

\begin{abstract}
Bayesian inference for graphical models has received much attention in the literature in recent years. It is well known that when the graph $G$ is decomposable, Bayesian inference is significantly more tractable than in the general non-decomposable setting. Penalized likelihood inference on the other hand has made tremendous gains in the past few years in terms of scalability and tractability. Bayesian inference, however, has not had the same level of success, though a scalable Bayesian approach has its respective strengths, especially in terms of quantifying uncertainty.  To address this gap, we propose a scalable and flexible novel Bayesian approach for estimation and model selection in Gaussian undirected graphical models. We first develop a class of generalized $G$-Wishart distributions with multiple shape parameters for an arbitrary underlying graph. This class  contains the $G$-Wishart distribution as a special case. We then introduce the class of Generalized Bartlett (GB) graphs, and derive an efficient Gibbs sampling algorithm to obtain posterior draws from generalized $G$-Wishart distributions corresponding to a GB graph. The class of Generalized Bartlett graphs contains the class of decomposable graphs as a special case, but is substantially larger than the class of decomposable graphs. We proceed to derive theoretical properties of the proposed Gibbs sampler. We then demonstrate that the proposed Gibbs sampler is scalable to significantly higher dimensional problems as compared to using an accept-reject or a Metropolis-Hasting algorithm. Finally, we show the efficacy of the proposed approach on simulated and real data.

 \vspace{0.5cm}
\end{abstract}

\section{Introduction}

\noindent
Gaussian graphical models have found widespread use in many application areas. Besides standard penalized likelihood based approaches (see \cite{kharesangraja2013} and references therein), Bayesian methods have also been proposed in the literature for analyzing undirected Gaussian graphical models \citep[see][to name just a few]{ascipiccioni2007, dawidlauritzen1993, letacmassam, mitsakakis2011, rajamassamcarv, roverato2000, roverato2002, carvalhowang}. Bayesian methods have the distinct and inherent advantage that they can incorporate prior information and yield a full posterior for the purposes of uncertainty quantification (and not just a point estimate),  whereas standard frequentist approaches for uncertainty quantification (such as the bootstrap) may be computationally burdensome and/or break down in high dimensional settings. However, it is well known that Bayesian methods for graphical models in high dimensional settings lag severely behind their regularized likelihood based counterparts, in the sense that they are not scalable except under restrictive assumptions on the underlying sparsity pattern (such as for decomposable graphs). Hence a scalable and more general approach to graphical models, with theoretical and computational safeguards, is critical to leveraging the advantages of posterior inference. 

To outline the issues with current Bayesian methods more clearly, consider {\it i.i.d.} vectors ${\bf Y}_1, {\bf Y}_2, \cdots, {\bf Y}_n$ drawn from a $p$-variate normal distribution with mean vector ${\bf 0}$ and a sparse inverse covariance matrix $\Omega$. The sparsity pattern in $\Omega$ can be encoded in terms of a graph $G$ on the set of variables as follows. If the variables $i$ and $j$ do not share an edge in $G$, then $\Omega_{ij} = 0$. Hence, an undirected (or concentration) graphical model corresponding to $G$ restricts the inverse covariance matrix $\Omega$ to a submanifold of the cone of positive definite matrices (referred to as $\mathbb{P}_G$). A Bayesian statistical analysis of these models requires specification of a prior distribution (supported on $\mathbb{P}_G$) for $\Omega$. \cite{dawidlauritzen1993} introduced a class of prior distributions for $\Sigma = \Omega^{-1}$ called the Hyper Inverse Wishart (HIW) distributions. The induced class of prior distributions for $\Omega$ (supported on $\mathbb{P}_G$) is known as the class of $G$-Wishart distributions (see  \cite{roverato2000}). This class of prior distributions is quite useful and popular, and has several desirable properties, including the fact that it corresponds to the Diaconis-Ylvisaker class of conjugate priors for the concentration graph model corresponding to the graph $G$.

Closed form computations of relevant quantities corresponding to the $G$-Wishart distribution, such as expected value of the precision matrix and quantiles, are in general available only if the underlying graph $G$ is decomposable, i.e., $G$ does not have any induced cycle of length greater than or equal to 4. A variety of approaches have been developed in the literature to generate samples from the $G$-Wishart distribution corresponding to a general non-decomposable graph. \cite{ascipiccioni2007}  have developed a maximal clique based Markov Chain Monte Carlo (MCMC) approach to sample from the $G$-Wishart distribution corresponding to a general graph $G$. \cite{lenkoski2013} develops a direct sampler for $G$-Wishart distributions corresponding to a general graph $G$. This approach uses an iterative algorithm to minimize an objective function over the space of positive definite matrices with appropriate sparsity constraints. \cite{carvalhowang} have developed an accept-reject algorithm to generate direct samples from the $G$-Wishart distribution corresponding to a general graph $G$. \cite{mitsakakis2011} have developed a Metropolis-Hastings based MCMC approach for the same. 

While the $G$-Wishart prior is clearly very useful for Bayesian inference in graphical models, it has an  important drawback. In particular, the $G$-Wishart distribution has only one shape parameter, which makes it potentially inflexible and restrictive in terms of prior specification. \cite{letacmassam} address this issue by constructing the so-called $W_{\mathbb{P}_G}$ and $W_{Q_G}$ families of distributions which are flexible in the sense that they have multiple shape parameters. These distributions include the $G$-Wishart as a special case, and form a standard conjugate family of prior distributions for undirected decomposable graphical models. The construction of the Letac and Massam distributions uses the structure associated with decomposable graphs. It would thus be useful to develop a class of prior distributions which is flexible (multiple shape parameters) and leads to tractable Bayesian inference for non-decomposable graphs.

In this paper, we aim to develop a scalable and flexible Bayesian approach for estimation and model selection in Gaussian undirected graphical models for general graphs. Our approach preserves the attractive properties of previous approaches, while overcoming their drawbacks. We first develop a class of generalized $G$-Wishart distributions (for an arbitrary underlying graph), which has multiple shape parameters and contains the $G$-Wishart distributions as a special case. These distributions form a family of standard conjugate prior distributions for Gaussian concentration graph models. Developing methods for efficient posterior draws from generalized $G$-Wishart distributions is crucial for scalable Bayesian inference. We proceed to introduce the class of Generalized Bartlett (GB) graphs, and derive an efficient Gibbs sampling algorithm (with Gaussian or GIG conditionals) to simulate from generalized $G$-Wishart distributions corresponding to a GB graph. The class of Generalized Bartlett graphs contains decomposable graphs as a special case, but is substantially larger than the class of decomposable graphs. For example, any cycle of length greater than $3$ is Generalized Bartlett, but is not decomposable. Our approach has the flexibility of using multiple shape parameters (as opposed to the single parameter $G$-Wishart), but goes beyond the class of decomposable graphs without losing tractability. 

For the generalized $G$-Wishart case, the conditional densities for any maximal clique of $\Omega$ are intractable to sample from. Hence, the sampling approaches in \citep{ascipiccioni2007, lenkoski2013} for $G$-Wisharts on a general graph do not extend to the generalized $G$-Wishart. On the other hand, we show that the accept-reject and Metropolis-Hastings based methods in \cite{carvalhowang} and \cite{mitsakakis2011} can be easily extended to the generalized $G$-Wishart case. We compare the performance and scalability of these two approaches with our Gibbs sampler in Section \ref{accept-reject} and Section \ref{metropolis-hastings}.

The rest of the paper is organized as follows. Section \ref{preliminaries} contains a brief overview of relevant concepts from graph theory and matrix theory. In Section \ref{gengwishart} and Section \ref{genbargraph}, we define generalized $G$-Wishart distributions and GB graphs respectively, and establish some basic properties. In Section \ref{gibbssampler}, we derive a tractable Gibbs sampling algorithm to simulate from the generalized $G$-Wishart distribution corresponding to a GB graph.  Section \ref{examplegenbar} provides additional examples and properties of GB graphs. Section \ref{illusandapp} contains a comprehensive simulation and real data analysis study for the Bayesian approach developed in the paper.  The proofs of most of the technical results in the paper and additional numerical work are provided in the Supplemental Document.

\section{Preliminaries} \label{preliminaries}

\subsection{Graph theoretic preliminaries}
\label{graphintro}
\noindent
For any positive integer $p$, let $\mathbb{N}_p := \{1,2, \cdots, p\}$. Let $G = (V, E)$ denote an undirected graph, where $V$ represents the finite vertex set and $E \subseteq V \times V$ denotes the corresponding edge set. A function $\sigma$ is defined to be an ordering of $V$ if $\sigma$ is a bijection from $V$ to $\mathbb{N}_{|V|}$. An undirected graph $G = (V,E)$ and an ordering $\sigma$ of $V$ can be used to construct an ordered graph $G_\sigma = (V,\sigma, E_\sigma)$, where $(i,j) \in E_\sigma$ if and only if $(\sigma^{-1} (i), \sigma^{-1} (j)) \in E$. 

\begin{Def}
An undirected graph $G=(V,E)$ is called \textbf{decomposable} if it 
does not have a cycle of length greater than or 
equal to $4$ as an induced subgraph. 
\end{Def}

\noindent
Such graphs are also called \textbf{triangulated}, or \textbf{chordal} 
graphs. A useful concept associated to decomposable graphs is that of 
a perfect elimination ordering (see \cite{lauritzen}). 
\begin{Def}
An ordering $\sigma$ for an undirected graph $G = (V,E)$ is defined to be a {\bf perfect elimination ordering} if 
for each $j \in N_{|V|}$, the set $\{j\} \cup \{i: i > j, (i,j) \in E_\sigma\}$ forms a clique. 
\end{Def}

\noindent
In fact, an undirected graph $G$ is decomposable if and only if it has a perfect elimination ordering (see \cite{paulsenetal1989}). 
\begin{Def}
For a given undirected graph $G=(V,E)$, $\widetilde{G}=(V,\widetilde{E})$ is called a \textit{decomposable cover} 
of $G$ if $\widetilde{G}$ is decomposable and $E \subset \widetilde{E}$. 
\end{Def}

\noindent
Decomposable covers are also known as {\it triangulations} in graph theory literature (see \cite{andreas}). 

\subsection{Matrix theoretic preliminaries}
\label{matrixtheory}
\noindent
We denote the set of $p \times p$ symmetric matrices by $\mathbb{M}_p$, and the space of $p \times p$ positive definite 
symmetric matrices by $\mathbb{M}_p^+$. Given an ordered graph $G_\sigma$, we define 
$$
\mathbb{P}_{G_\sigma} = \{\Omega \in \mathbb{M}_{|V|}^+: \; \Omega_{ij} = 0 \mbox{ if } (i,j) \notin E_\sigma\}, 
$$

\noindent
and 
$$
\mathcal{L}_{G_\sigma} = \{L \in \mathbb{M}_{|V|}: \; L_{ii} = 1, \; L_{ij} = 0 \mbox{ for } i < j \mbox{ or } 
(i,j) \notin E_\sigma\}. 
$$

\noindent
The space $\mathbb{P}_{G_\sigma}$ is a submanifold of the space of $|V| \times |V|$ positive definite matrices, 
where the elements are restricted to be zero whenever the corresponding edge is missing from $E_\sigma$. 
Similarly the space $\mathcal{L}_{G_\sigma}$ is a subspace of lower triangular matrices with diagonal entries 
equal to $1$, such that the elements in the lower triangle are restricted to be zero whenever the corresponding 
edge is missing from $E_\sigma$. 

A positive definite matrix {\it $\Omega$} can be uniquely expressed as $\Omega = LDL^T$, where $L$ is a lower triangular matrix with diagonal entries equal to $1$, and $D$ is a diagonal matrix with positive diagonal entries. Such a decomposition is known as the {\it modified Cholesky decomposition} of $\Omega$ (see for example \cite{danielspourahmadi}). \cite{paulsenetal1989} showed that if $\Omega \in \mathbb{P}_{G_\sigma}$, then $L \in \mathcal{L}_{G_\sigma}$ if and only if $G$ is decomposable and $\sigma$ is a perfect elimination ordering. If either of these two conditions is violated, then the sparsity pattern in $L$ is a strict subset of the sparsity pattern in $\Omega$. The entries $(i,j) 
\notin E_\sigma$ (with $i > j$) such that $L_{ij}$ is not (functionally) zero, are known as ``fill-in" entries. The problem of finding an ordering which minimizes the number of fill-in entries is well-known and well-studied in numerical analysis and in computer science/discrete mathematics. Although this problem is NP-hard, several effective greedy algorithms for reducing the number of fill-in entries have been developed and implemented in standard software such as MATLAB and R (see \cite{Davis} for instance). 

In subsequent sections, we will consider a reparametrization from (the inverse covariance) matrix $\Omega$ to its modified Cholesky decomposition. Such a reparametrization inherently assumes an ordering of the variables. In many applications (such as longitudinal data), a natural ordering is available. In the absence of a natural ordering, one can choose a fill-reducing ordering using one of the available fill-reducing algorithms mentioned previously. We will see that a fill-reducing ordering will help in reducing the computational complexity of proposed Markov chain Monte Carlo procedures.

\subsection{Undirected graphical models and $G$-Wishart distribution}

\noindent
Let $G = (V,E)$ be an undirected graph with $|V| = p$, and $\sigma$ be an ordering of $V$. The undirected graphical model corresponding to the the ordered graph $G_\sigma$ is the family of distributions 
$$
\mathcal{J} = \{MVN_p ({\bf 0}, \Omega^{-1}): \Omega \in \mathbb{P}_{G_\sigma}\}. 
$$
\noindent
Let $\bm{Y_1},\bm{Y_2},\ldots,\bm{Y_n}$ be {\it i.i.d.} observations from a distribution in $\mathcal{J}$. Note that the joint density of $\bm{Y_1},\bm{Y_2},\ldots,\bm{Y_n}$ given $\Omega$ is given by
\[\frac{|\Omega|^{\frac{n}{2}}}{(\sqrt{2\pi})^{np}} \exp{\left(-\frac{n}{2} tr(\Omega S)\right)}.\]

\noindent
The $G$-Wishart distribution on $\mathbb{P}_{G_\sigma}$ is a natural choice of prior for $\Omega$ (see \cite{dawidlauritzen1993} and \cite{roverato2000}). The density of the $G$-Wishart distribution with parameters $\delta > 0$ and $U \in \mathbb{M}^+_p$ is proportional to
\[|\Omega|^{\frac{\delta}{2}}\exp{\left(-\frac{1}{2}tr(\Omega U)\right)}.\]
\noindent
Thus the posterior density of $\Omega$ given $\bm{Y_1},\bm{Y_2},\ldots,\bm{Y_n}$ is proportional to
\[|\Omega|^{\frac{n+\delta}{2}}\exp{\left(-\frac{1}{2}tr(\Omega(nS+U))\right)},\]
\noindent
and corresponds to a $G$-Wishart distribution with parameters $(n+\delta)$ and $(U+nS)$, which implies that the family of $G$-Wishart priors are conjugate for the family of distributions $\mathcal{J}$.

\section{Generalized G-Wishart distributions}
\label{gengwishart}
\noindent
In this section we propose a generalization of the $G$-Wishart distribution that is endowed with multiple shape parameters, and contains the $G$-Wishart family as a 
special case. We shall show in later sections that the flexibility offered by the multiple shape parameters is very useful in high dimensional settings. 

\subsection{Definition}
\noindent
We now define a multiple shape parameter generalization of the $G$-Wishart distribution for a general graph $G$. To do this, we transform the matrix $\Omega$ to its Cholesky decomposition. Consider the modified Cholesky decomposition $ \Omega = LDL^T$, where $L$ is a lower triangular matrix with diagonal entries equal to $1$, and $D$ is a diagonal matrix with positive diagonal entries. The (unnormalized) density of the {\it generalized $G$-Wishart distribution} with parameters 
$\boldsymbol{\delta} = (\delta_1, \delta_2, \cdots, \delta_p) \in \mathbb{R}_+^p$ and $U \in 
\mathbb{M}_p^+$ is given by 
\begin{equation} \label{defggwshrt}
\pi_{U, \boldsymbol{\delta}}^* (\Omega) = \left( \prod_{i=1}^p D_{ii} (\Omega)^{\frac{\delta_i}{2}} \right) 
\exp{\left(-\frac{1}{2}tr(\Omega U)\right)}. 
\end{equation}

\noindent
We note that other generalizations of the Wishart have also been considered in \citet{bendavidrajaratnam,danielspourahmadi, dawidlauritzen1993, khareraja,letacmassam}. It is clear that the $G$-Wishart density arises as a special case of the generalized $G$-Wishart (by considering all the $\delta_i$'s to be equal and noting that $|\Omega| = \prod_{i=1}^p D_{ii}$), and that the family of generalized $G$-Wishart distributions defined above is a conjugate family of prior distributions for undirected graphical models. In fact, the posterior density of $\Omega$ corresponds to a generalized $G$-Wishart distribution with parameters $n+\delta_1, \cdots, n + \delta_p$ and $(U+nS)$. 

\subsection{Some properties of the generalized G-Wishart distribution}
\label{propofgengwish}
We now proceed to derive properties of the generalized $G$-Wishart distribution. To do so, we transform $\Omega$ to its modified Cholesky decomposition $\Omega=LDL^T$ as defined in Section \ref{matrixtheory}.

We define $L_I = \{L_{ij}| i>j\mbox{ and }(i,j) \in E_{\sigma}\}$ to be the set of functionally independent elements of $L$. Then the transformation $\Omega \rightarrow (L_I,D)$ is a bijection from $\mathbb{P}_{G_{\sigma}}$ to $\mathbb{R}^{\frac{|E_{\sigma}|}{2}} \times \mathbb{R}_+^p$ with Jacobian equal to $\prod_{j=1}^p D_j^{\nu_j}$, where $\nu_j := |\{i:i>j,(i,j) \in E_{\sigma}\}|$ for $j=1,2,\ldots,p$. Then the (unnormalized) generalized $G$-Wishart density for $(L_I,D)$ is given by 
\begin{equation}
\pi_{U, \boldsymbol{\delta}}^* (L_I, D) = \left(\prod_{j=1}^p D_j^{\frac{\delta_j  +2\nu_j}{2}}\right) \times \exp{\left(-\frac{1}{2}\sum_{i=1}^p\sum_{j=1}^p
U_{ij} \sum_{k=1}^{\min(i,j)} L_{ik}L_{jk}D_k\right)}
\label{defpistar} 
\end{equation}

\noindent
We first establish sufficient conditions for the density $\pi_{U, \boldsymbol{\delta}}^*$ to be proper. 
\begin{thm} \label{thm1}
If $U$ is positive definite and $\delta_i>0\,\, \forall i=1,\ldots,p$, 
\[\int \pi_{U,\boldsymbol{\delta}}^*(L_I,D) d(L_I,D) < \infty.\]
Also under these conditions, $E\left[\Omega_{ij}\right] < \infty$, $\forall i,j$.
\end{thm}

\noindent
The proof of Theorem \ref{thm1} is provided in Supplemental Section \ref{proofthm1}. Under the conditions in Theorem \ref{thm1}, we will refer to the normalized version of 
$\pi_{U,\boldsymbol{\delta}}^*$ as $\pi_{U,\boldsymbol{\delta}}$. 

From \cite{roverato2000}, if $\Omega$ follows $G$-Wishart with parameters $(U,
\delta)$ then for $(i,j) \in E$ or $i=j$,
\[E((\Omega^{-1})_{ij}) = \frac{U_{ij}}{\delta}\]

\noindent
We now provide an extension of this result for the case of generalized $G$-Wishart 
distributions. 
\begin{thm} \label{thm2}
Let $\Omega=LDL^T \in \mathbb{P}_{G_\sigma}$ be generalized $G$-Wishart with parameters $(U,\boldsymbol{\delta})$ for some $U \in M_P^+,\boldsymbol{\delta} \in \mathbb{R}_+^p$. Denote  $\Omega_k$ as the $k \times k$ principal submatrix of $\Omega$, and let $[\Omega_k^{-1}]^0$ denote the $p \times p$ matrix with $\Omega_k^{-1}$ as its appropriate $k \times k$ principal submatrix and rest of the elements equal to $0$. Define the matrix $U_{G_\sigma}$ as $(U_{G_\sigma})_{ij}=U_{ij}\times1_{\{(i,j) \in E_\sigma\}}$. If $\delta_k > 4, \forall k$, then
\[E \left[ \sum_{k \leq p} \left(\delta_k-\delta_{k+1} \right) [\Omega_k^{-1}]^0 
\right] = U_{G_\sigma}. 
\]
\end{thm}

\noindent
The proof of Theorem \ref{thm2} is quite detailed and technical and is thus provided in Supplemental Section \ref{proofthm2}. Theorem \ref{thm2} 
provides a useful tool to monitor convergence of any Markov chain Monte Carlo 
method for sampling from $\pi_{U, \boldsymbol{\delta}}$ (and particularly the Gibbs sampler introduced in Section \ref{gibbssampler}) . 

We also undertake a comparison between generalized $G$-Wishart distribution and the useful priors introduced by \cite{letacmassam} for the case of decomposable graphs. A careful analysis demonstrates that the generalized $G$-Wishart and Letac-Massam priors are quite different for decomposable graphs. The generalized $G$-Wishart coincides with the Letac-Massam Type II Wishart in the special case when $G$ is homogeneous. See Supplemental Section \ref{completacmassam} for details.

\section{Generalized Bartlett graphs}
\label{genbargraph}
As discussed earlier, the class of decomposable graphs is endowed with many properties helpful for closed form computation of posterior quantities. The assumption of decomposability can be rather restrictive in higher dimensions, as they constitute a very small fraction of all graphs see (Figure \ref{prop-propq}). We develop in this section a class of graphs, which is substantially larger than the class of decomposable graphs. We will show in later sections that for this class of graphs, we can generate posterior samples from the generalized G-Wishart, using a tractable Gibbs sampling algorithm.

\subsection{Preliminaries and Definitions}
\noindent
We now provide the definition of Generalized Bartlett graphs. First consider the following procedure to obtain a decomposable cover (see Section \ref{graphintro}) of an arbitrary ordered graph $G_{\sigma}=(V,E_\sigma)$.

\begin{algorithm}
\caption{Triangulation Algorithm for an unordered graph $G$}
\begin{algorithmic}[0]
   \State Denote $G_0^{\sigma}:=G$.
   \While{$i \leq p-2$} 
      \State $E_i^{\sigma} = E_{i-1}^{\sigma} \cup \{(u,v)|\sigma(u) > \sigma(v) > i \in \mathbb{N}_{|V|} \mbox{ \& } (u,\sigma^{-1}(i)),(v,\sigma^{-1}(i)) \in E_{i-1}^{\sigma} \}$
      \State $G_i^{\sigma} = (V, E_i^{\sigma})$
   \EndWhile
   \label{triangulation}
\end{algorithmic}
\end{algorithm}

\noindent
We use the above algorithm to construct a decomposable cover for $G$ as follows. Let $D^\sigma (G) := G_{p-2}^{\sigma}$, and let $D^\sigma (E)$ denote the edge set of $D^\sigma (G)$. It follows by construction that the ordering $\sigma$ is a perfect vertex elimination scheme for $D^\sigma (G)$. Hence, $D^{\sigma}(G)$ is a decomposable cover for $G$. Note that, two different orderings may give rise to different decomposable covers. We now define Generalized Bartlett graphs. 

\begin{Def}
An undirected graph $G=(V,E)$ is said to be a Generalized Bartlett graph if there exists an ordering $\sigma$, with the property that there does not exist vertices $u,v,w \in V$ satisfying $(u,v),(v,w),(u,w) \notin E$ and $(u,v),(v,w),(u,w) \in D^\sigma (E)$. 
\end{Def}
\noindent
In such a case (i.e., when this property is satisfied), $\sigma$ is called a \textit{Generalized Bartlett 
ordering} of $G$. When it exists, the Generalized Bartlett ordering may not be unique. The following lemma helps in proving an alternate characterization of Generalized Bartlett graphs, one that does not depend on any ordering of the vertices. The lemma shows that Algorithm \ref{triangulation} leads to a collection of minimal decomposable covers, in the sense that the edge set of any decomposable cover of $G$ has to contain $D^\sigma (E)$ for some ordering $\sigma$. 

\begin{lem} \label{lem1}
For any undirected graph $G=(V,E)$ and a decomposable cover $\widetilde{G}=(V,\widetilde{E})$ of $G$, $\exists$ an ordering $\sigma$ of $V$ s.t., $D^{\sigma}(E) \subset \widetilde{E}$.
\end{lem}
\begin{proof}
Let $\widetilde{G}=(V,\widetilde{E})$ be a decomposable cover of $G=(V,E)$. Since $\widetilde{G}$ is decomposable let $\sigma$ be the perfect elimination ordering of it. We shall prove inductively that for $i$ in $\{0,1,\ldots,p-2\}$, $E_i^{\sigma} \subset \widetilde{E}$. Since $D^{\sigma}(E)=E^{\sigma}_{p-2}$, that will prove this lemma. It is trivial to note that $E=E^{\sigma}_0 \subset \widetilde{E}$. Let us assume that the claim holds for $E^{\sigma}_{i-1}$. Consider any $r>s>i$ such that $(\sigma^{-1}(r),\sigma^{-1}(s)) \in E^{\sigma}_i \setminus E^{\sigma}_{i-1}$. Thus $(\sigma^{-1}(r),\sigma^{-1}(i)),(\sigma^{-1}(s),\sigma^{-1}(i)) \in E^{\sigma}_{i-1} \subset \widetilde{E}$. Since $\sigma$ is the perfect elimination ordering for $\widetilde{G}$, $(\sigma^{-1}(r),\sigma^{-1}(s)) \in \widetilde{E}$. Thus $E^{\sigma}_i \subset \widetilde{E}$ which completes the induction step and proves the lemma. 
\end{proof}
\noindent
We now provide a second definition of \textit{Generalized Bartlett graphs}. 
\begin{lem} \label{lem2}
An undirected graph $G$ satisfies the Generalized Bartlett property if and only if $G = (V,E)$ has a decomposable cover $\widetilde{G} = (V, \widetilde{E})$ such that every triangle in $\widetilde{E}$ contains an edge from $E$. That is for any $u,v,w \in V$ such that $(u,v),(v,w),(u,w) \in \widetilde{E}$, at least one of $(u,v),(v,w),(u,w)$ belongs to $E$.
\end{lem}
\begin{proof}

If $G$ satisfies the Generalized Bartlett property then by definition $\exists$ an ordering $\sigma$ of $V$ such that any triangle in $D^{\sigma}(E)$ contains an edge from $E$. In that case we can simply take $\widetilde{E}=D^{\sigma}(E)$. On the other hand, let $\widetilde{G} = (V, \widetilde{E})$ be a decomposable cover of $G$ such that every triangle in $\widetilde{E}$ contains an edge from $E$. By Lemma \ref{lem1}, $\exists$ an ordering $\sigma$ of $V$, such that $D^{\sigma}(E) \subset \widetilde{E}$. Thus any triangle in $D^{\sigma}(E)$ is also a triangle in $\widetilde{E}$ and hence has an edge in $E$. This makes $\sigma$ the Generalized Bartlett ordering and $G$ an Generalized Bartlett graph.

\end{proof}
In the subsequent arguments we will also refer to Generalized Bartlett graphs as satisfying the Generalized Bartlett property. Some common Generalized Bartlett graphs are: 
\begin{enumerate}
\item All decomposable graphs (follows by using $\widetilde{G} = G$ in Lemma \ref{lem2}). 
\item Any cycle (see Section \ref{cycleexmpl} for a proof). This is the simplest example of a non-decomposable graph which satisfies the Generalized Bartlett property. Also note that $4$ cycle is the simplest non-decomposable graph.
\item Any lattice with less than $4$ rows or less than $4$ columns (see Section \ref{gridsexmpl} for a proof). Such graphs are useful in spatial applications (see Section \ref{climate}). 
\end{enumerate}

\begin{figure}[h]
\includegraphics[scale=1]{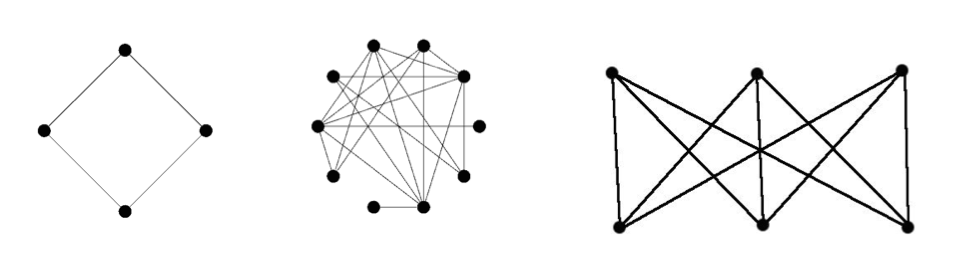}
\caption{(left and middle) Examples of Generalized Bartlett graphs. (right) The bipartite graph $K_{3,3}$: the only 6-vertex connected graph which is not Generalized Bartlett.}
\label{6bipartite}
\end{figure}

It is a natural question to ask how much larger the class of Generalized Bartlett graphs is compared to the class of decomposable graphs. It is quite difficult to obtain a closed form expression for the exact (or approximate) number of connected decomposable graphs (or Generalized Bartlett graphs) with a given number of vertices. However, a list of all possible connected non-isomorphic graphs having at most $10$ vertices is available at \url{http://cs.anu.edu.au/~bdm/data/graphs.html}. Using this list, we computed the number of decomposable and Generalized Bartlett graphs with at most $10$ vertices. Figure \ref{prop-propq} provides a graphical comparison of these proportions, and Table \ref{table} gives the actual values of these proportions. It is quite clear that the proportion of Generalized Bartlett graphs is much larger than the proportion of decomposable graphs. As expected, the proportions of both classes of graphs decreases as the total number of vertices increases. However, the rate of decrease in the proportions is much larger for decomposable graphs. For example, less than $0.02 \%$ of connected isomorphic graphs with $10$ vertices is decomposable, but more than $85 \%$ of connected isomorphic graphs with $10$ vertices is Generalized Bartlett. In this case the number of Generalized Bartlett graphs are approximately $100$ times larger.

\begin{figure}[h] 
\centering
\begin{subfigure}{.4\textwidth}
\includegraphics[scale=0.4]{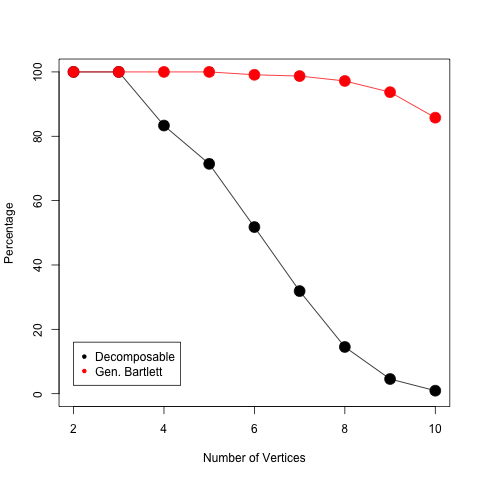}
\caption{Plot comparing the percentage\\ of Generalized Bartlett graphs with \\decomposable graphs}
\label{prop-propq}
\end{subfigure} \hspace{1cm}
\begin{subfigure}{.4\textwidth}
\captionsetup{justification=raggedleft,
singlelinecheck=false
}
\includegraphics[scale=0.5]{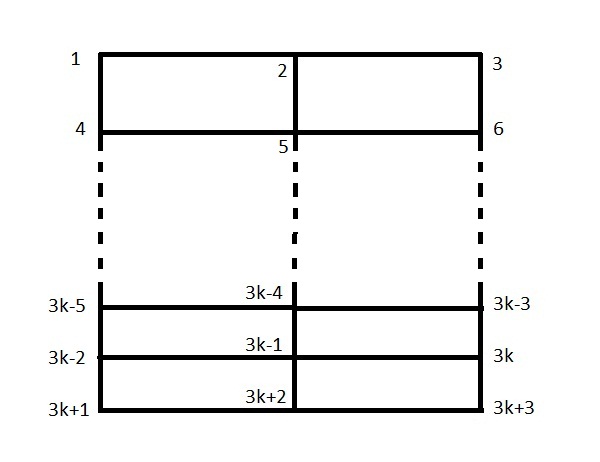}
\caption{A $(k+1) \times 3$ grid which is a \hspace{0.2cm} \\ Generalized Bartlett graph.}
\label{3kby3}
\end{subfigure}
\caption{ }
\end{figure}

\begin{figure}[h]
\begin{subfigure}{0.5\textwidth}
\begin{tabular}{|l|p{1.5cm}|p{1cm}|p{1cm}|p{1cm}|}
\hline
 Order & No. of graphs & \multicolumn{2}{|c|}{Percentage} & Ratio \\ \hline
 & & Decom posable & Gen. Bart. &  \\ \hline
2 & 1 & 100 & 100 & 100\% \\
3 & 2 & 100 & 100 & 100\% \\
4 & 6 & 83 & 100 & 83\% \\
5 & 21 & 71 & 100 & 71\% \\
6 & 112 & 52 & 99 & 53\% \\
7 & 853 & 32 & 98 & 32\% \\
8 & 11117 & 15 & 97 & 15\% \\
9 & 261080 & 4.5 & 94 & 5\% \\ 
10 & 11716571 & 0.9 & 86 & 1\% \\ \hline
\end{tabular}
\caption{Percentages for Generalized Bartlett graphs and decomposable graphs among connected 
non-isomorphic graphs with at most $10$ vertices.}
\label{table}
\end{subfigure}\hfill%
\begin{subfigure}{0.4\textwidth}
\centering
\includegraphics[scale=0.3]{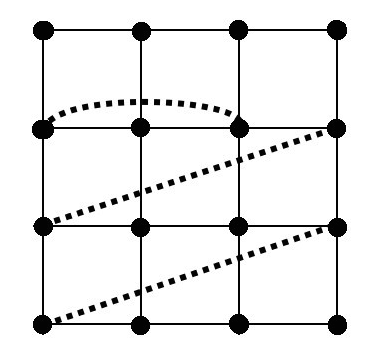}
\caption{A $4 \times 4$ grid where the dotted lines represent the extra edges in its Generalized Bartlett cover.}
\label{4by4}
\end{subfigure}
\caption{ }
\end{figure}

\subsection{Clique Sum Property of Generalized Bartlett graphs}
An unordered graph $G=(V,E)$ is said to have a decomposition into components $G_1=(V_1,E_1)$ and $G_2=(V_2,E_2)$ if the vertex set $V$ can be decomposed as $V=V_1 \cup V_2$ where $(V_1-V_2) \cup (V_2-V_1) \neq \emptyset$ such that the induced subgraph on $V_1 \cap V_2 \neq \emptyset$ is complete, and $V_1 \cap V_2$ separates $V_1-V_2$ from $V_2-V_1$ (i.e., if $u \in V_1-V_2$ and $w \in V_2-V_1$ then $(u,w) \notin E$). If $G$ cannot be decomposed in the above manner it is called a \textit{prime graph}. Hence any graph $G$ is either prime or can be broken down into several prime components by repeated application of the above procedure. It is well known that a graph $G$ is decomposable iff all of its prime components are complete. The following lemma provides a similar characterization for Generalized Bartlett graphs.
\begin{lem}
If all the prime components of a graph are Generalized Bartlett then the graph is also Generalized Bartlett.
\end{lem}
\begin{proof}
Note that, it is enough to prove the theorem, for a graph with two prime components. Let $G=(V,E)$ be an 
undirected graph with $V=V_1 \cup V_2$ such that induced subgraph on $V_1 \cap V_2 \neq \emptyset$ is complete, and 
$V_1 \cap V_2$ separates $V_1-V_2$ and $V_2-V_1$. Let $G_1=(V_1,E_1)$ and $G_2=(V_2,E_2)$ be the corresponding induced 
subgraphs which are Generalized Bartlett. 
Then by Lemma 2, we can construct decomposable covers $\widetilde{G}_1=(V_1,\widetilde{E}_1)$ and $\widetilde{G}_2=(V_2,
\widetilde{E}_2)$ of $G_1$ and $G_2$ respectively, such that every triangle in $\widetilde{E}_i$ contains an edge in 
$E_i$ for $i = 1,2$. Define $\widetilde{E} = \widetilde{E}_1 \cup \widetilde{E}_2$ and $\widetilde{G}=(V,\widetilde{E})$. 
Note that $\widetilde{G}$ is a cover for $G$. We claim that $\widetilde{G}$ is decomposable. Suppose to the contrary, 
that for for some $n \geq 4$ there exists an induced cycle in $\widetilde{G}$ on a set of $n$ vertices, say $u_1,u_2,
\ldots,u_n$. Since $\widetilde{G}_1$ and $\widetilde{G}_2$ are both decomposable, all of $\{u_1,u_2,\ldots,u_n\}$  cannot 
belong to exclusively in $V_1$ or in $V_2$. Hence, there exist $1 \leq i,j \leq n$ such that $u_i \in V_1-V_2$ and $u_j 
\in V_2-V_1$. Since $V_1-V_2$ and $V_2-V_1$ are separated by $V_1 \cap V_2$, $(u_i,u_j) \notin \widetilde{E}$. Thus the 
subgraph induced by $\widetilde{G}$ on the set of vertices $\{u_1, u_2, \cdots, u_n\}$ contains two paths, both arising from $u_i$ and ending in $u_j$ and intersecting no where in between. Let $\{u_i,u_{i_1},\ldots,u_{i_p},u_j\} $ be one of those paths. If  $\{u_i,u_{i_1},\ldots,u_{i_p},u_j\} \subset (V_1-V_2)\cup V_2-V_1$, then there exists points in $V_1-V_2$ and $V_2-V_1$ connected to each other in $\widetilde{G}$, which is not possible. Thus $\exists \, u_k \in \{u_{i_1},\ldots,u_{i_p}\}$ s.t. $u_k \in V_1 \cap V_2$. Similarly $\exists \, u_{k'} \in V_1 \cap V_2$ corresponding to the second path. Since $V_1 \cap V_2$ is complete $(u_k,u_{k'}) \in \widetilde{E}$. Hence $(u_k,u_{k'})$ is a chord in the cycle $u_1,u_2,\ldots,u_n$ giving us a contradiction. Hence, $\widetilde{G}$ is a decomposable cover of $G$.
 
To prove that $G$ is a Generalized Bartlett graph we shall prove that every triangle in the decomposable cover $\widetilde{G}$ contains an edge in $G$. Let us assume to the contrary that for $u,v,w \in V$, $(u,v),(u,w),(v,w) \in \widetilde{E}$ but $(u,v),(u,w),(v,w) \notin E$. Since every triangle in $\widetilde{G}_i$ has at least one edge 
from $G_i$ for $i = 1,2$, it follows that $u,v,w$ all cannot belong exclusively to $V_1$ or $V_2$. Without loss of generality, let $u \in V_1-V_2$ and $v \in V_2-V_1$. This implies $(u,v) \notin \widetilde{E}$ giving us a contradiction. Thus $G$ is Generalized Bartlett.
\end{proof}

\subsection{Constructing Generalized Bartlett covers}
Given an ordered graph $G_\sigma =(V,E_\sigma)$, Algorithm \ref{genbarcover} below provides a Generalized Bartlett graph $\underline{G}_\sigma = (V, \underline{E}_\sigma)$ such that $\underline{E}_\sigma \supset E_\sigma$. Such a graph is referred to as a Generalized Bartlett cover of $G_\sigma$.

\begin{algorithm}[h]
\caption{Construction of Generalized Bartlett cover}
\begin{algorithmic}[0]
   \State Set $\tilde{E}:=E$.
   \While{$\exists \, i>j,(\sigma^{-1}(i),\sigma^{-1}(j)) \notin \tilde{E}$ such that $L_{ij}$ when expressed as a polynomial violates Property A or B} 
      \State $\tilde{E}=\tilde{E} \cup \{(\sigma^{-1}(i),\sigma^{-1}(j))\}$
   \EndWhile   
\end{algorithmic}
\label{genbarcover}
\end{algorithm}

Recall from Section \ref{gridsexmpl} that a $4 \times 4$ grid is the smallest example of a grid which is not a Generalized Bartlett graph. The Generalized Bartlett cover for a $4 \times 4$ grid using Algorithm \ref{genbarcover} is provided in Figure \ref{4by4}. Note that this cover has only three extra edges.

\section{A tractable Gibbs sampler for generalized $G$-Wisharts}
\label{gibbssampler}

\noindent
In Section \ref{genbargraph} we studied the graph theoretic properties of GB graphs. In this section we shall investigate the statistical/properties of GB graphs. In particular, we develop a tractable Gibbs sampling algorithm to sample from the generalized $G$-Wishart distribution when the underlying graph is Generalized Bartlett. The first step in achieving this goal requires considering a further transformation of the Cholesky parameter $(L_I, D)$ from Section \ref{propofgengwish}.

\subsection{A reparametrization of the Cholesky parameter}

\noindent
Let $G = (V, E)$  be an undirected graph with $|V| = p$, and $\sigma$ an ordering for $G$. Let $\Omega = LDL^T$ be the modified Cholesky decomposition of $\Omega \in \mathbb{P}_{G_\sigma}$. To facilitate our analysis, we consider a one-to-one transformation of $(L_I, D)$ defined as follows:
\[(D_1,D_2,\ldots,D_p) \rightarrow (\widetilde{D}_1,\widetilde{D}_2,\ldots,\widetilde{D}_p)\]
where $\widetilde{D}_1 = D_1$ and $\widetilde{D}_k = \frac{D_k}{D_{k-1}}$ for $2 \leq k \leq p$. The 
following lemma shows that terms of the form $L_{ik} L_{jk} D_k$ can be expressed as a polynomial in the entries of $L_I$ and $\widetilde{D}$ (with negative powers allowed for entries of $\widetilde{D}$). 
\begin{lem} \label{lem:reparametrization}
For $1 \leq k \leq j \leq i \leq p$, terms of the form $L_{ik} L_{jk} D_k$, which appear in the modified Cholesky expansion of $\Omega$, are either functionally zero, or can be expressed as a sum of terms, where each term has the following form:
\begin{equation}
\pm \left(\prod_{\{r>s,(r,s) \in E, r\leq i, s<j\}}L_{rs}^{c'_{rs}}\right) \times \left(\prod_{k=1}^p \widetilde{D}_k^{d'_k}\right) 
\label{term2}
\end{equation}
\end{lem}

\noindent
{\it Proof}: Note that \[D_i = \prod_{k=1}^i \widetilde{D}_k\] for all $i$, and the Jacobian of this 
transformation is $\prod_{k=1}^{p-1} \widetilde{D}_k ^{p-k}$. Hence, the posterior density of $(L_I, 
\widetilde{D})$ is proportional to 
\begin{equation} \label{posterior1}
\left(\prod_{j=1}^p \widetilde{D}_j^{\alpha_j}\right) \times \exp{\left(-\frac{1}{2} \sum_{i=1}^p 
\sum_{j=1}^p(nS_{ij}+U_{ij})\sum_{k=1}^{\min(i,j)} L_{ik}L_{jk} \prod_{l=1}^k \widetilde{D}_l \right)}, 
\end{equation}

\noindent
where $\alpha_k = (p-k)+ \sum_{l=k}^p\frac{n+\delta_l+2n_l}{2}$ for $1 \leq k \leq p$. Note that if $i>j$ and $(i,j) 
\notin E$, then $\Omega_{ij} = \sum_{k=1}^j L_{ik} L_{jk} D_k = 0$, which implies 
$$
L_{ij} = -\frac{\sum_{k=1}^{j-1} L_{ik} L_{jk} D_k}{D_j} = -\sum_{k=1}^{j-1} L_{ik} L_{jk} \prod_{l=k+1}^j 
\widetilde{D}_l^{-1}. 
$$

\noindent
Making repeated substitutions in the RHS of the above equation, it follows that the entry $L_{ij}$ is either functionally zero, or can be expressed as a sum of terms, where each term looks like
\begin{equation}
\pm \left(\prod_{\{r>s,(r,s) \in E_\sigma\}}L_{rs}^{c_{rs}}\right) \times \left(\prod_{k=1}^p 
\widetilde{D}_k^{d_k}\right) 
\label{term1}
\end{equation}

\noindent
for suitable non-negative integers $c_{rs}$, and non-positive integers $d_k$. It is easy to see that 
\begin{eqnarray} \label{powerexpan}
c_{rs} &=& 0 \mbox{ for } (r,s) \in E \mbox{ with } r > i \mbox{ or } s \geq j, \nonumber\\
d_j &=& -1 \mbox{ and } d_k = 0 \mbox{ for } k > j. \nonumber 
\end{eqnarray}

\noindent
Hence, $L_{ij}$ can be expressed as a polynomial in entries of $L_I$ and $\widetilde{D}^{-1}$. The 
results now follows from $\eqref{term1}$. \hfill$\Box$ 

\medskip

\noindent
Note that for every $i > j$ with $(i,j) \notin E_\sigma$, the functionally dependent entry $L_{ij}$ can be expressed in terms of $L_I$ and $\widetilde{D}$ as in (\ref{term1}). The above analysis indicates that, in general, the posterior is a complicated function of $(L_I, \widetilde{D})$. However, we will show that if $G$ is a Generalized Bartlett graph, and $\sigma$ is a Generalized Bartlett ordering for $G$, then the full conditional posterior distributions of all individual entries of $(L_I, \widetilde{D})$ are either Gaussian or Generalized Inverse Gaussian distributions (and therefore easy to sample from). This property will then be used to derive a Gibbs sampling algorithm to sample from the posterior density in (\ref{posterior1}). 

\subsection{The Gibbs sampler}

\noindent
We now derive a Gibbs sampling algorithm to sample from the posterior density in (\ref{posterior1}). We start by defining two properties which will be crucial to the development of the Gibbs sampler. 

\begin{Def}
Let $G_\sigma=(V,\sigma,E_\sigma)$ be an ordered graph, and for $\Omega \in \mathbb{P}_{G_\sigma}$, $\Omega = LDL^T$ be the modified Cholesky decomposition. 
\begin{enumerate}
\item The ordered graph $G_\sigma$ is defined to have \textbf{Property-A} if for every $i > j$
such that $(i,j) \notin E_\sigma$, the following holds: for every $r > s$ with $(r,s) \in E_\sigma$, $L_{ij}$ is a linear 
function of $L_{rs}$ (keeping other entries of $L_I$ and $\widetilde{D}$ fixed). In other words, in $\eqref{term1}$,  
$c_{rs}$ can only be $0$ or $1$ for every $r > s$ with $(r,s) \in E_\sigma$. 
\item The ordered graph $G_\sigma$ is defined to have \textbf{Property-B} if for every $i > j$ 
such that $(i,j) \notin E_\sigma$, the following holds: for every $1 \leq k \leq p$, $L_{ij}$ is a linear function of $\widetilde{D}_k^{-1}$ (keeping other entries of $L_I$ and $\widetilde{D}$ fixed). In other words, in $\eqref{term1}$, $d_k$ can only be $0$ or $-1$ for every $1 \leq k \leq p$. 
\end{enumerate}
\end{Def}

\noindent
We now state three lemmas which will be useful in our analysis. The first lemma provides an equivalent 
formulation of Property-B. The proofs of these lemmas are provided in Supplemental Section \ref{proof4.0}, \ref{proofnocutting} and \ref{proofotherone} respectively. 
\begin{lem} \label{lem4.0}
The following statements are equivalent. 
\begin{enumerate}[(a)]
\item The ordered graph $G_\sigma$ satisfies Property-B. 
\item For every $1 \leq k \leq j \leq i \leq p$, $L_{ik} L_{jk} D_k$ can be expressed as a polynomial in entries of $L_I$ 
and $\widetilde{D}$ (with negative powers allowed for entries of $\widetilde{D}$). Furthermore for every term in the expansion of $L_{ik} L_{jk} D_k$, the power of any entry of $\widetilde{D}$ is either $0,1$ or $-1$. 
\end{enumerate}
\end{lem}

\noindent
Let $i > j$ with $(i,j) \notin E_\sigma$. As noted above, for $k' < k < j$, both $L_{ik} L_{jk} D_k$ and 
$L_{ik'} L_{jk'} D_{k'}$ can be expressed as polynomials in entries of $L_I$ and $\widetilde{D}$ (with negative 
powers allowed for entries of $\widetilde{D}$). 
\begin{lem} \label{nocutting}
Both $L_{ik} L_{jk} D_k$ and $L_{ik'} L_{jk'} D_{k'}$ are either functionally zero, or any term in the expansion of $L_{ik} L_{jk} D_k$ cannot be the exact negative (functionally) of any term in the expansion of $L_{ik'} L_{jk'} 
D_{k'}$. 
\end{lem}

\noindent
The next lemma shows that Generalized Bartlett graphs satisfies Properties A and B.
\begin{lem}
\label{genbarimpliesab}
Let $G=(V,E)$ be a Generalized Bartlett graph, and $\sigma$ be a Generalized Bartlett ordering for $G$. Then, the ordered graph $G_\sigma$ satisfies both Property A and B.
\end{lem}

\noindent
We are now ready to state and prove the main result of this section, which provides a Gibbs sampler for the posterior density in (\ref{posterior1}). 
\begin{thm}
Let $G=(V,E)$ be a Generalized Bartlett Graph, and $\sigma$ be a Generalized Bartlett ordering for $G$. Suppose $\Omega \in \mathbb{P}_{G_\sigma}$ follows a generalized $G$-Wishart distribution with parameters $U$ and $\boldsymbol{\delta}$ for some positive definite matrix $U$ and $\boldsymbol{\delta} >0$. If $\Omega = LDL^T$ is the modified Cholesky decomposition of $\Omega$, and we define $\widetilde{D}_1=D_1$,  $\widetilde{D}_k = \frac{D_k}{D_{k-1}}$, for $k \geq 2$, then 
\begin{itemize}
\item the conditional posterior density of the independent entry $L_{ij}$, given all other entries of $L_I$ and $\widetilde{D}$, is univariate normal,  
\item the conditional posterior density of $\widetilde{D}_k$ given $L_I$ and other entries of $\{\widetilde{D}_{k'}\}_{k' \neq k}$ is either a Generalized Inverse Gaussian or Gamma.
\end{itemize}
\label{megatheorem}
\end{thm}

\begin{proof}
Note that for $i > j$, $L_{ij}$ can be expressed as a polynomial in the entries of $L_I$ and $\widetilde{D}^{-1}$. Recall 
that $D^\sigma (G) = (V, D^\sigma (E))$ is the decomposable cover obtained by the triangulation process described in 
Algorithm \ref{triangulation}. We first establish that for $i>j$, $L_{ij}$ is functionally non-zero iff $(\sigma^{-1}(i),
\sigma^{-1}(j)) \in D^{\sigma}(E)$. We begin by noticing that at each step in the construction of $G_1,G_2,\ldots,G_{p-1}$ we are adding some extra edges to $G_{i-1}$ to get $G_i$, but never deducting anything. So if $L_{ij}$ was an independent entry, i.e. $(\sigma^{-1}(i),\sigma^{-1}(j)) \in E_0=E$, then $(\sigma^{-1}(i),\sigma^{-1}(j)) \in E_{p-2}=D^{\sigma}(E)$.

Now lets assume to the contrary that $L_{ij}$ is the first dependent but functionally non-zero entry s.t. $(\sigma^{-1}(i),\sigma^{-1}(j)) \notin D^{\sigma}(E)$. Since $L_{ij}$ is dependent but non-zero $\exists$ $k<j$ s.t. $L_{ik}L_{jk}\frac{D_k}{D_j} \neq 0$ appears in the expansion of $L_{ij}$. Now $L_{ik}$ can be independent and hence $(\sigma^{-1}(i),\sigma^{-1}(k)) \in E\subset E_k$. Otherwise $L_{ik}$ is non-zero dependent. Since $L_{ij}$(the first non-zero dependent not in $D^{\sigma}(E)$) comes after $L_{ik}$, $(\sigma^{-1}(i),\sigma^{-1}(k)) \in D^{\sigma}(E)$. We recall that, for any $l$, while constructing $G_l$, we only join vertices higher than $l$. Thus $(\sigma^{-1}(i),\sigma^{-1}(k))$ must have been joined before construction of $G_k$, i.e. $(\sigma^{-1}(i),\sigma^{-1}(k)) \in E_{k-1}$.
By a similar argument $(\sigma^{-1}(j),\sigma^{-1}(k)) \in E_{k-1}$. Thus $(\sigma^{-1}(i),\sigma^{-1}(j)) \in E_k \subset D^{\sigma}(E)$ and we have a contradiction. Hence,
\[L_{ij} \neq 0 \implies (\sigma^{-1}(i),\sigma^{-1}(j)) \in D^{\sigma}(E)\]
\noindent
To prove the reverse implication we note that for $r>s$, $(\sigma^{-1}(r),\sigma^{-1}(s)) \in E=E_0$ implies that $L_{rs}$ is independent and hence $\neq 0$. Now we use induction and assume that the claim holds upto $E_{i-1}$. If for $r>s>i$, $(\sigma^{-1}(r),\sigma^{-1}(s)) \in E_i\setminus E_{i-1}$ then  $(\sigma^{-1}(r),\sigma^{-1}(s)) \notin E$ and hence $L_{ri}L_{si}\frac{D_i}{D_s}$ appears in the expansion of $L_{rs}$. Since $(\sigma^{-1}(r),\sigma^{-1}(s)) \in E_i\setminus E_{i-1}$, $(\sigma^{-1}(r),\sigma^{-1}(i)),(\sigma^{-1}(s),\sigma^{-1}(i)) \in E_{i-1}$ and by assumption $L_{ri} \neq 0$ and $L_{si} \neq 0$ which with the help of Lemma \ref{nocutting} implies $L_{rs}\neq 0$. Thus the assumption holds for $E_i$, which completes the induction step.

It follows by (\ref{posterior1}) and Property-A that for every $i>j,(i,j) \in E$, the conditional posterior density of $L_{ij}$ given all other entries of $L_I$ and $\widetilde{D}$ is proportional to
\[\exp{(-a_{ij}(L_{ij}-b_{ij})^2)},\]

\noindent
for appropriate constants $a_{ij}$ and $b_{ij}$. Hence the conditional posterior density of $L_{ij}$ is a Gaussian density. Similarly, it follows from (\ref{posterior1}) and Property-B that for every $1 \leq k \leq p$, the 
conditional posterior density of $\widetilde{D}_{k}$ given all entries of $L_I$ and $\{\widetilde{D}_{k'}\}_{k' \neq 
k}$ is proportional to 
\[\widetilde{D}_{k}^{\alpha_k} \exp{\left(-\widetilde{a}_k \widetilde{D}_{k}-\frac{\widetilde{b}_k}{\widetilde{D}_{k}}\right)}\]

\noindent
for appropriate constants $\widetilde{a}_k$ and $\widetilde{b}_k$. Hence the conditional posterior density of $
\widetilde{D}_1$ is Gamma, and for $k\geq 2$, the conditional posterior density of $\widetilde{D}_k$ is a Generalized 
Inverse Gaussian density. 
\end{proof}

\noindent
The results in Theorem \ref{megatheorem} can be used to construct a Gibbs sampling algorithm, where the iterations involve sequentially sampling from the conditional densities of each element of $(L_I,D)$. It is well known that the joint posterior density of $(L_I,\widetilde{D})$ is invariant for the Gibbs transition density. Since the Gaussian 
density is supported on the entire real line, and the Generalized Inverse Gaussian density is supported on the entire positive real line, it follows that the Markov transition density of the Gibbs sampler is strictly positive. Hence, the corresponding Markov chain is aperiodic and $\lambda$-irreducible where $\lambda$ is the Lebesgue measure on $\mathbb{R}^{|L_I|}\times \mathbb{R}_+^p$(\cite{meyntweedie}, Pg 87). Also, the existence of an invariant probability density together with $\lambda$-irreducibility imply that the chain is positive Harris recurrent (see \cite{asmussenglynn} for instance). We formalize the convergence of our Gibbs sampler below. The following lemma on the convergence of the Gibbs sampling Markov chain facilitates computation of expected values for generalized $G$-Wishart distributions.
\begin{lem}
Let $G = (V,E)$ be a Generalized Bartlett graph, and $\sigma$ be a Generalized Bartlett ordering for $G$. Then, the 
Markov chain corresponding to the Gibbs sampling algorithm in Theorem \ref{megatheorem} is positive Harris recurrent.
\end{lem}

\subsection{Maximality of Generalized Bartlett graphs}

\noindent
Note that the Gibbs sampling algorithm described in Theorem \ref{megatheorem} is feasible only if Property-A and Property-B hold. The following theorem shows that if a graph is not Generalized Bartlett, then at least one of Property-A and Property-B does not hold.
\begin{thm} \label{maximality}
If an ordered graph $G_\sigma$ satisfies Property-A and Property-B, then the graph $G$ is a Generalized Bartlett graph and $\sigma$ is a Generalized Bartlett ordering for $G$. 
\end{thm}

\begin{proof}
Suppose there exists $i>j>k$ s.t.$(i,j),(i,k),(j,k) \notin E$ but $(i,j),(i,k),(j,k) \in D^{\sigma}(E)$ i.e. $L_{ij}\neq0,L_{ik}\neq0,L_{jk}\neq0$. Hence $L_{ik}L_{jk}\frac{D_k}{D_j}$  is in the expansion of $L_{ij}$. The power of $\widetilde{D}_k$ in the expansion of $L_{ik}$ and $L_{jk}$ is $-1$. $\frac{D_k}{D_j} = \widetilde{D}_{k+1}^{-1} \ldots \widetilde{D}_{j}^{-1}$. Also we know that, no term of $L_{ik}L_{jk}\frac{D_k}{D_j}$ can cancel with any term of $L_{ik'}L_{jk'}\frac{D_k'}{D_j}$. Hence in the expansion of $L_{ij}$ the power of $\widetilde{D}_{k}$ is $-2$. Thus Property-B is violated. The result now follows by Definition \ref{genbarimpliesab}.\\
\end{proof}
\noindent
Theorem \ref{maximality} demonstrates that the class of Generalized Bartlett graphs is maximal, in the sense that the conditional distributions considered in Theorem \ref{megatheorem} are Gaussian/Generalized-Inverse-Gaussian only if the underlying graph is Generalized Bartlett. In other words the above tractability is lost for graphs outside the Generalized Bartlett class.

\subsection{Improving efficiency using decomposable subgraphs}

\noindent
It is generally expected that `blocking' or `grouping' improves the speed of convergence of Gibbs samplers (see \cite{liu:wong:kong:1994}). Suppose $\Omega = LDL^T$ follows a generalized $G$-Wishart distribution. In this section, we will show that under appropriate conditions, the conditional density of a block of variables in $(L_I, D)$ (given the other variables) is multivariate normal. Based on the discussion above, this result can be used to sample more efficiently from the joint density of $(L_I, D)$. 
\begin{lem}
Let $G_\sigma=(V,E_\sigma)$ be a Generalized Bartlett graph with $p$ vertices and $\Omega(=LDL^T)$ 
follows generalized $G$-Wishart with parameters $(U,\boldsymbol{\delta})$. Suppose that for some $1<p_1<p$, the induced subgraph of $G_\sigma$ corresponding to the vertices $\{p_1+1,\ldots,p\}$ is decomposable with a perfect elimination ordering. Then $\{L_{ij}|p_1<j<i \leq p,(i,j) \in E_\sigma\}|(L_I \setminus \{L_{ij}|p_1<j<i \leq p,(i,j) \in E_\sigma\},D)$ follows a multivariate normal distribution. 
\end{lem}

\begin{proof}
We partition the matrix $L$ as
\[L = \left(\begin{array}{cc}
L_1 & 0\\
L_2 & L_3\\
\end{array} \right)\]

\noindent
where $L_1$ has dimension $p_1 \times p_1$ and correspondingly, 
\[U = \left(\begin{array}{cc}
U_1 & U_2^T\\
U_2 & U_3\\
\end{array} \right), \, D = \left(\begin{array}{cc}
D_1 & 0\\
0 & D_2\\
\end{array} \right).\]

\noindent
Note that the density of $(L_I,D)$ is proportional to 
\[\prod_{i=1}^p D_{ii}^{\nu_i+\delta_i/2} \exp \left(-\frac{1}{2} tr(LDL^T U) \right). \]

\noindent
A sample calculation gives:
\begin{equation} \label{dcmtrgraph}
tr(LDL^T U) = tr(L_1D_1L_1^T U_1)+2 tr(L_1D_1L_2^T U_2) +  tr(L_2D_1L_2^T U_3) + tr(D_2 L_3^T U_3 L_3). 
\end{equation}

\noindent
Consider $i > i' > p_1$ such that $(i,i') \notin E_\sigma$. Since the induced subgraph on 
$\{p_1 + 1, \cdots, p\}$ is a decomposable graph with a perfect elimination ordering, 
there does not exist $p_1 < j < i'$ such that $(i,j), (i',j) \in E_\sigma$. It follows 
that $L_{ii'} = -\sum_{j=1}^{i'-1} L_{ij} L_{i'j} D_{j}/D_{i'}$ is a function of entries 
in $(L_1, L_2, D)$. Hence, all the dependent entries in $L_3$ are functions of 
$(L_1, L_2, D)$. It follows by (\ref{dcmtrgraph}) that given $(L_1, L_2, D)$, $tr \left( 
LDL^T U \right)$ is a quadratic form in the independent entries of $L_3$. Hence, the log 
of the conditional density of $\{L_{ij}|p_1<j<i \leq p, (i,j) \in E_\sigma\}$ given the other entries in $(L_I, D)$ is a quadratic form. This proves the required result.
\end{proof}


\subsection{Closed form expressions for decomposable graphs}
\label{decomposableclosedform}
A closed form expression for the mean of $\Omega$ can be obtained if $G$ is assumed to be a decomposable graph.
For $\{\delta_i>0|i=1,\ldots,p\}$ and $U$ positive definite, the generalized $G$-Wishart density on $\mathbb{P}_{G_\sigma}$ is,
\[\pi_{U,\boldsymbol{\delta}}(\Omega) \propto \prod_{j=1}^p D_j^{\frac{\delta_j}{2}} \exp \left(-\frac{1}{2} tr(\Omega U) \right) \mbox{ where } \Omega \in \mathbb{P}_{G_\sigma}\]

Let us define,

\begin{table}[h]
\begin{tabular}{ll}
$\mathcal{N}^{\succ j} := \{i:i>j,(i,j)\in E_\sigma\}$ & \hspace{1cm} $\nu_j := |\mathcal{N}^{\succ j}|$\\
$\boldsymbol{U^{\succ}_{.j}}:=\{U_{ij}|i \in \mathcal{N}^{\succ j}\}$ &\hspace{1cm} $U^{\succ j} := \{U_{ii'}|i,i' \in \mathcal{N}^{\succ j}\}$ \\
$\boldsymbol{e_j} := -(U^{\succ j})^{-1} \bm{U^{\succ}_{.j}}$ &\hspace{1cm} $c_j := U_{jj}-\bm{U^{\succ T}_{.j}} (U^{\succ j})^{-1} \bm{U^{\succ}_{.j}}>0$ \\
\end{tabular}
\end{table}

Also let $\mathcal{H}$ be the diagonal matrix with $(k,k)$-th element as $\frac{\delta_k + \nu_k +2}{c_k}$ and $e$ is a $p \times p$ matrix whose $(k,j)$-th element is $e_{k,j}$ if $j \in \mathcal{N}^{\succ k}$, is $1$ if $j=k$ and $0$ otherwise. The following theorem provides closed form expectations of the elements of the matrix $\Omega$.

\begin{thm}
\label{thmdecomposableclosedform}
If $G_\sigma$ is decomposable where the vertices have been ordered by an perfect elimination ordering, and $\Omega=LDL^T \in \mathbb{P}_{G_\sigma}$ is generalized $G$-Wishart with parameters $(U,\boldsymbol{\delta})$, then
\begin{equation}
L_{I_j}|D_j \sim  N\left( \boldsymbol{e_j},\frac{(U^{\succ j})^{-1}}{D_j} \right) \, \mbox{ and } \,
D_j \sim Gamma \left( \frac{\nu_j+\delta_j}{2}+1,\frac{c_j}{2} \right)
\end{equation}

where $L_{I_j}$ are the independent entries of the $j$-th column of $L$ and $D=diag(D_1,D_2,\ldots,D_p)$.
Also,
\[ E(\Omega) = \sum_{k<p}[(U^{\succ k})^{-1}]^0 +  e^T \mathcal{H} e \]
\end{thm}
The proof is provided in the Supplemental Section \ref{proofthmdecomposableclosedform}.

\section{Classes of Generalized Bartlett Graphs}
\label{examplegenbar}
As mentioned earlier, the class of Generalized Bartlett graphs contains the class of decomposable graphs. In this section we will consider two naturally occurring examples of non-decomposable Generalized Bartlett graphs. We then provide schemes for combining a group of Generalized Bartlett  graphs to produce a bigger Generalized Bartlett graph. 

\subsection{The $p$-cycle} \label{cycleexmpl}

\noindent
We show that the $p$-cycle (with its standard ordering) satisfies Property-A and Property-B, and is hence  a Generalized Bartlett graph. Let $G_\sigma=(V,\sigma,E_{\sigma})$, where $V=N_p$, $\sigma$ is the identity permutation and 
$$E_{\sigma}=\{(i,j): 1 \leq i,j \leq p, |i-j| \in \{1,p-1\}\}.$$  

\noindent
The independent entries of $L$ are $L_{21},L_{32},\ldots,L_{p(p-1)},L_{p1}$. After some straightforward algebraic 
manipulations, the dependent entries of $L$ can be calculated as follows: 
\begin{eqnarray*}
L_{ij} &=& 0 \hspace{6cm} \mbox{ if } i \neq p,i>j+1 \\
&=&  (-1)^{j-1} L_{p1} \left(\prod_{k=2}^j L_{k(k-1)}\right) \frac{D_1}{D_j}  \hspace{1cm} \mbox{ if } i=p \mbox{ and } 2 \leq j \leq p-2\\
\end{eqnarray*}

\noindent
It is clear from the expressions in the above equation that Property-A and Property-B are satisfied and hence by Theorem 4, the $p$-cycle is Generalized Bartlett. 

\subsection{Grids} \label{gridsexmpl}

A $m \times n$ grid is an undirected graph formed by the intersection of $m$ rows and $n$ columns where the vertices correspond to the $p=mn$ intersection points and as a result $m*(n-1)+n*(m-1)$ edges are formed.  In this section we shall prove that for some particular ordering all $n \times 2$ and $n \times 3$ grid are Generalized Bartlett. We order an $n \times 3$ grid row wise starting from the top as shown in Figure \ref{3kby3}.

Let $G_\sigma = (V,E_\sigma)$ be an ordered graph, $\Omega \in \mathbb{P}_{G_\sigma}$, and $\Omega = LDL^T$ denote the modified Cholesky decomposition of $\Omega$. Note that Property-A and Property-B have been defined for ordered graphs, but we extend these notions to polynomials as follows. For any polynomial $p(L_I,\widetilde{D})$ of $(L_I,\widetilde{D})$, we say that $p()$ satisfies Property-A if the power of any independent $L_{ij}$ can be $\{0,1\}$. Similarly we say $p()$ satisfies Property-B if the power of any $\widetilde{D}_k$ can be $\{-1,0\}$. We note that if $j<i,i'$ and $L_{ij}$ satisfies Property-B then $L_{ij}\frac{D_j}{D_{i'}}$ also satisfies Property-B.

\begin{lem} \label{lemgrid}
An $n \times 3$ grid when ordered as above is Generalized Bartlett.
\end{lem}


\subsection{Expansion property of Generalized Bartlett graphs}
In this section we develop two methods, which combine an arbitrary number of Generalized Bartlett graphs in a suitable manner to produce a larger Generalized Bartlett graph.

\subsubsection{Maximum vertex based expansion}
We start by proving a lemma which will be useful for further analysis.

\begin{lem}
Let $G=(V,E)$ be a Generalized Bartlett graph. If $V'$ is a subset of $V$ and $G'=(V',E')$ is the corresponding subgraph then $G'$ is also a Generalized Bartlett graph.
\end{lem}

\begin{proof}
Since $G$ is Generalized Bartlett, let $\underline{G}$ be the decomposable cover of $G$ such that any triangle in $\underline{G}$ has at least one edge in $G$. Let $\underline{G'}$ be the induced subgraph of $\underline{G}$ for $V'$. Then $\underline{G'}$ is decomposable since it is a induced subgraph of a decomposable graph. Also any triangle in $\underline{G'}$ is a triangle in $\underline{G}$ and thus has atleast one edge in $G$ and hence in $G'$. Thus by Lemma \ref{lem2}, $G'$ is Generalized Bartlett. Moreover from the proof of Lemma \ref{lem2} we can observe that if $\sigma$ is the Generalized Bartlett ordering for $G$ then the same ordering works for $G'$. 
\end{proof}

\noindent Let $G=(V,E)$ be a Generalized Bartlett graph with, $V=\{1,2, \ldots, r\}$. Suppose we  replace each vertex $i$ of $G$, by a Generalized Bartlett graph $G_i=(V_i,E_i)$, where for $i=1,2,\ldots,r$, $V_i = \{p_{i-1}+1, \ldots ,p_i\}$. Here $p_0=0$ and $p_1,p_2,\ldots,p_r$ are the sizes of the $r$ graphs. Note that the graphs being considered here are already ordered. For ease of exposition, we will suppress the ordered graph notation, and refer to the graphs as just $G,G_1,G_2,\ldots$.

\begin{Def}
The expanded graph $\widetilde{G}=(\widetilde{V},\widetilde{E})$ is constructed from $G$ using $G_1,G_2,\ldots,G_r$  as follows,
\begin{itemize}
\item $\widetilde{V} = V_1 \cup V_2 \cup \ldots \cup V_r = \{1,2,\ldots,p_r\}$
\item $(k,l) \in \widetilde{E}$ iff \textbf{either} $(k,l) \in E_i$ for some $G_i$, \textbf{or} $k=p_i,l=p_j$  for some $1 \leq i \neq j \leq r$ and $(i,j) \in E$.
\end{itemize}
\end{Def}

\noindent
Hence $\widetilde{G}$ is constructed from $G$ by replacing the $i$-th vertex of $G$ by $G_i$. An edge between $i$ and $j$ in $G$ translates to an edge between the maximal vertices of $G_i$ and $G_j$ namely $p_i$ and $p_j$.  For any $i=1,2, \ldots, p_r$, if $p_{k-1}+1 \leq i \leq p_k$, the notation  $Graph(i)$ shall denote $G_k$, and $Graph\_before(i)$ shall denote $\cup_{s<k} G_s = (\cup_{s<k}V_s, \cup_{s<k}E_s)$ 

\begin{thm}
The expanded graph $\widetilde{G}$ defined as above is Generalized Bartlett.
\label{expansion}
\end{thm}
\noindent
The proof this theorem is given in the Supplemental Section \ref{proofexpansion}.

\subsubsection{Tree based expansion}
Consider a tree $T$ with $r$ vertices $\{v_1,v_2,\ldots,v_r\}$. For each $v_i$ consider an arbitrary number of GB graphs say $G_{v_i}^{(1)},\ldots,G_{v_i}^{(n_i)}$. We add an edge from $v_i$ to each vertex in $G_{v_i}^{(j)}$ for every $1 \leq j \leq n_i$. Denote the resulting graph by $G=(V,E)$. Next the vertices in $V$ are labeled in the following order. 
\[G_1^{(1)},\ldots,G_1^{(n_1)},G_2^{(1)},\ldots,G_2^{(n_2)},\ldots,G_r^{(1)},\ldots,G_r^{(n_r)},T\]
The labeling is done in such a way that the induced ordering on each $G_{i}^{(j)}$ is a GB ordering and every parent vertex in $T$ gets a higher label than any of its children in $T$. Again for ease of exposition we will suppress the ordering notation and refer to the resulting ordered graph as $G=(V,E)$.

\begin{thm}
\label{tree}
The graph $G$ defined and ordered as above is Generalized Bartlett.
\end{thm}
\noindent
The proof of this theorem is provided in the Supplemental Section \ref{prooftree}.

\section{Illustrations and Applications}
\label{illusandapp}
\noindent
We now illustrate the advantages of our Generalized Bartlett approach on both simulated and real data and  demonstrate that the proposed GB method is scalable to significantly higher dimensions. In Section  \ref{gwishvsgengwish}, we illustrate the advantage of having multiple shape parameters in the generalized $G$-Wishart distribution. In Section \ref{accept-reject} and Section \ref{metropolis-hastings}, we undertake a comparison of our algorithm with the accept-reject and Metropolis-Hastings approaches. Section \ref{climate} contains a real data analysis using data from a temperature study. Although the main focus of this paper is development of the flexible class of $G$-Wishart distributions, and tractable methods to sample from these distributions, we also illustrate that the methods developed in this paper can be used for high-dimensional graphical model selection in conjunction with existing penalized likelihood methods  (see Supplemental Sections \ref{modelselectionexp} and \ref{breastcancerdata}).

\subsection{Comparing $G$-Wishart with generalized $G$-Wisharts}
\label{gwishvsgengwish}

In this section, we present a simulation experiment to demonstrate that the multiple shape parameters in the generalized  $G$-Wishart distribution can yield differential shrinkage and improved estimation as compared to the single parameter $G$-Wishart in higher dimensional setting.

For the purposes of this experiment we consider a Generalized Bartlett graph $G$ with $p=1000$ vertices, defined as follows. Let,
\begin{eqnarray*}
b_1 &=& 50, b_2 = 150, b_3 = 450, b_4 = 1000 \\
B_1 &=& \{1,\ldots, 49\}, B_2 = \{51, \ldots ,149\}, B_3 = \{151, \ldots, 449\},B_4 = \{451,\ldots, 999\} 
\end{eqnarray*}
\noindent
A graph $G$ is constructed by forming the $4$-cycle $\{b_1,b_2,b_3,b_4\}$ and then connecting $b_i$ with all elements of $B_i$ for $i=1,2,3,4$. An inverse covariance matrix $\Omega_0 \in \mathbb{P}_G$ is then constructed by taking $\Omega_0=L_0 D_0 L_0^T$, where $(D_0)_{jj}=b_i-b_{i-1}$ if $j \in \{b_i\} \cup B_i$. Here $L_0$ is a lower triangular matrix with independent entries equal to $0.5$, and dependent entries chosen such that $\Omega_0 \in \mathbb{P}_G$. 

We then generate $n=100$ samples from a $N(\boldsymbol{0},\Sigma_0=\Omega_0^{-1})$ distribution. Let $S$ denote the corresponding sample covariance matrix. Let $c$ denote the mean of the diagonal entries of $n*S$. We first consider a $G$-Wishart prior for $\Omega$ with $U=c I_p$ and different choices of $\delta$. Using the Gibbs sampler proposed in Section \ref{gibbssampler}, the posterior mean for $\Omega$ (and $\Sigma$) is then computed for each choice of $\delta$. Figure \ref{toolcover} depicts the performance of these posterior mean estimators in terms of the Steins loss function (denote by $L_1$). It can be seen from Figure \ref{toolcover} that, $L_1(\hat{\Omega},\Omega_0)$ and $L_1(\hat{\Sigma},\Sigma_0)$ are  minimized at $\delta=262$ and $\delta=353$ respectively.

\begin{figure}
\includegraphics[scale=0.05]{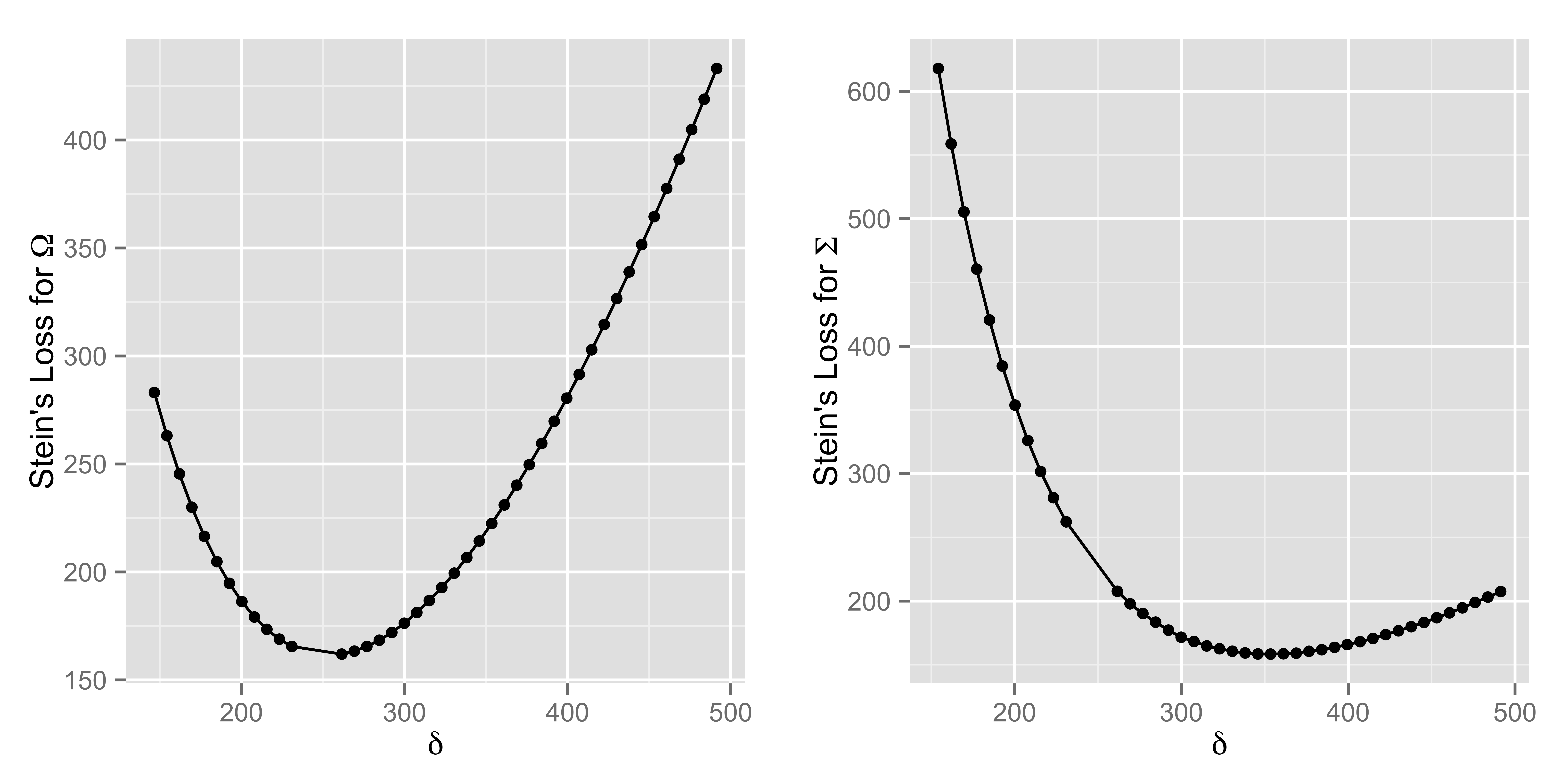}
\caption{Performance of $G$-Wishart prior for different values of $\delta$. The x-axis represents the chosen value of $\delta$. The y-axis represents Stein's loss between the estimated posterior and the true value for $\Omega$ and $\Sigma$ respectively.}
\label{toolcover}
\end{figure}

We now illustrate the improved performance of the posterior mean estimators when using our generalized $G$-Wishart priors endowed with multiple shape parameters. If $\Omega = \Sigma^{-1}$ follows a $G$-Wishart distribution with parameters $U$ and $\delta > 0$, then for every $i = j$ and $(i,j) \in E$, $E[\Sigma_{ij}] = U_{ij}/\delta$ see \cite[see][Corollary 2]{roverato2002}. Borrowing intuition from this result, we first choose a generalized $G$-Wishart empirical prior for $\Omega$ with $U=cI_p$ and $\boldsymbol{\delta}=diag(U + nS)/diag(S)$. Here $diag(\cdot)$ denotes the vector of diagonal entries of a given matrix. It can be seen from Table \ref{4values} that even with this empirical choice of $\boldsymbol{\delta}$ we observe a $30.2\%$ decrease in Stein's loss for $\Omega$ and $13\%$ for $\Sigma$ compared to the best performance in the single shape parameter case. Next, we perform a restricted grid search to check if the performance can be further improved. In particular, a $4$-dimensional grid search is performed on $(m_1,m_2,m_3,m_4)$ where for we assign $\delta_l=m_i$ for all $l \in B_i \cup \{b_i\}$. 
As shown in Table \ref{4values}, the best posterior mean estimator obtained via this search improves the Stein's loss for $\Omega$ and $\Sigma$ by $35\%$ and $27.7\%$ respectively (compared to the best estimator in the single parameter case).  

\begin{table}[h]
\centering
\begin{tabular}{|l|c|c|}
\hline
$\boldsymbol{\delta}$ & $L_1(\hat{\Omega},\Omega_0)$ & $L_1(\hat{\Sigma},\Sigma_0)$\\ \hline
$262*1_p$ & 161.9 & 207.6 \\
$353*1_p$ & 222.4 & 158.3 \\ \hline
$diag(U^*)/diag(S)$ &  113.0  & 137.7 \\
$m_1=140,m_2=140,m_3=264,m_4=348$ &  105.2  &  114.4\\ \hline
\end{tabular}
\caption{The first two rows indicate the best possible performance using single shape parameter. The next two rows indicate the performance for an objective choices with multiple shape parameters and an attractive choices obtained by doing a grid search over certain 4-valued $\boldsymbol{\delta}$ vector.}
\label{4values}
\end{table}

\subsection{Comparison with other Monte Carlo based approaches}

\noindent
We shall show in this section that two other approaches, namely the accept-reject algorithm and the Metropolis algorithm, can also be used to sample from the generalized $G$-Wishart distribution. We demonstrate however that both these algorithms can have specific scalability issues, but the proposed Gibbs sampler can overcome these challenges.  

\subsubsection{Comparison with the accept-reject algorithm}
\label{accept-reject}
A useful accept-reject algorithm to simulate $\Omega$ from a $G$-Wishart distribution for a general graph $G$ is provided in \cite{carvalhowang}. This algorithm can be easily generalized to simulate from our generalized $G$-Wishart distributions. Below, we present this generalized version of the accept-reject algorithm for a graph $G$, where $G$ cannot be decomposed into any prime components. Many graphs, such as cycles, $m \times n$ grids (with $m,n \geq 3$) satisfy this property. 

Let $\Omega \in \mathbb{P}_{G_\sigma}$ follow a generalized $G$-Wishart distribution with parameters $U$ and $\boldsymbol{\delta}$ for some positive definite matrix $U$ and $\boldsymbol{\delta} \in \mathbb{R}_+^p$. Let $U^{-1}=T'T$ be the Cholesky decompositions of $U^{-1}$. Also for $i = \{1,\ldots,p\}$, define $\nu_i := |\{s|s>i,\,(s,i) \in E_{\sigma} \}|$. The accept-reject algorithm can now be specified as follows. 
\begin{description}
\item[\underline{Step 1}] Simulate $\Psi$ as follows;
For $i \in \sigma(V)$, $\psi_{ii}^2 \sim \chi^2_{\delta_i+\nu_i+2}$ and for $i<j,\,(i,j)\in E_{\sigma}$, $\psi_{ij} \sim N(0,1)$. For $i<j,\,(i,j) \notin E_{\sigma}$, we calculate $\psi_{ij}$ as,
\begin{eqnarray*}
\psi_{ij} &=& - \sum_{k=1}^{j-1} \psi_{1k} T_{kj}/T_{jj} \mbox{\hspace{10cm} if } i=1 \\
&=& -\sum_{k=i}^{j-1} \psi_{ik} T_{kj}/T_{jj} - \sum_{k=1}^{i-1} \left( \frac{\psi_{ki}+\sum_{l=k}^{i-1} \psi_{kl}T_{l,i-1}/T_{i-1,i-1}}{\psi_{ii}}\right) \left(\psi_{kj}+\sum_{l=k}^{j-1} \psi_{kl}T_{l,j-1}/T_{j-1,j-1} \right) \\
&& \hspace{13cm} \mbox{ if } 1<i<j\\
\end{eqnarray*}

\item[\underline{Step 2}] Simulate $u \sim U(0,1)$ and check whether 
\[u < \exp \left(-\frac{1}{2}\sum_{i<j,\, (i,j)\notin E_{\sigma}} \psi_{ij}^2 \right).\]
If this holds, then accept this value of $\Psi$, else go to \textbf{Step 1}.

\item[\underline{Step 3}] Set $\Phi = \Psi T$ and $\Omega = \Phi'\Phi$.Then $\Omega$ has the required 
generalized Wishart distribution. 
\end{description}


\noindent
A common problem with the vanilla application of accept-reject algorithm even in moderate dimensional settings is that the average acceptance probability can be extremely small. This issue can make the accept-reject algorithm computationally infeasible. We find that the same phenomenon happens with the accept-reject algorithm in the generalized $G$-Wishart distribution setting. The algorithm works well for small dimensional examples, such as the simulation example in \cite{carvalhowang} for a $7$-vertex graph (where the largest prime component has order $4$). We find however, that the low acceptance probability issue mentioned above, surfaces as we increase the size of the largest prime component. To illustrate this, let $G$ be a $12$ cycle (which cannot be decomposed into prime components) with an ordering $\sigma$ as specified in Section 4. Consider a generalized $G$-Wishart distribution with parameters $U^0$ and $\boldsymbol{\delta}$. Here $U^0_{ii} = 100$ and for $|i-j|=1,11$ $U^0_{ij}=40$ and $0$ otherwise. The shape parameters, $\delta_i=60$ for $1 \leq i \leq 6$, and $\delta_i = 70$ for $7 \leq i \leq 12$. The average time taken by the accept-reject algorithm to complete one iteration is more than $5$ hours on a 2.4 Ghz processor with 4 GB RAM. Clearly, this happens due to low acceptance probabilities. However, the Gibbs sampler does not suffer from these issues, since no acceptance/rejection step is involved. The results obtained by using $10000$ iterations of the Gibbs sampler (which take approximately $4$ minutes) are provided in Table \ref{performance}. In particular, this table demonstrates that the difference between mean of $\Sigma^*_{ij}$ from the Gibbs sampler (for $i=j$ or $(i,j) \in E_\sigma$) and its theoretical expectation $U^0 [i,j]$ (as given by Theorem \ref{thm2}) is very small.
 
\begin{table}
\begin{tabular}{|c|c|c|c|c|c|c|c|}
\hline
i & j & Simulated mean   & True Mean & i & j & Simulated mean  & True Mean \\
\hline
$ 1 $ & $ 1 $ & $ 100.6 $ & $ 100 $ & $ 2 $ & $ 1 $ & $ 39.7 $ & $ 40 $ \\ 
$ 2 $ & $ 2 $ & $ 98.99 $ & $ 100 $ & $ 3 $ & $ 2 $ & $ 40.14 $ & $ 40 $ \\ 
$ 3 $ & $ 3 $ & $ 100.4 $ & $ 100 $ & $ 4 $ & $ 3 $ & $ 40.25 $ & $ 40 $ \\ 
$ 4 $ & $ 4 $ & $ 100.3 $ & $ 100 $ & $ 5 $ & $ 4 $ & $ 39.95 $ & $ 40 $ \\ 
$ 5 $ & $ 5 $ & $ 100 $ & $ 100 $ & $ 6 $ & $ 5 $ & $ 39.79 $ & $ 40 $ \\ 
$ 6 $ & $ 6 $ & $ 99.82 $ & $ 100 $ & $ 7 $ & $ 6 $ & $ 39.75 $ & $ 40 $ \\ 
$ 7 $ & $ 7 $ & $ 99.3 $ & $ 100 $ & $ 8 $ & $ 7 $ & $ 40.23 $ & $ 40 $ \\ 
$ 8 $ & $ 8 $ & $ 100.6 $ & $ 100 $ & $ 9 $ & $ 8 $ & $ 39.97 $ & $ 40 $ \\ 
$ 9 $ & $ 9 $ & $ 100.1 $ & $ 100 $ & $ 10 $ & $ 9 $ & $ 39.9 $ & $ 40 $ \\ 
$ 10 $ & $ 10 $ & $ 99.73 $ & $ 100 $ & $ 11 $ & $ 10 $ & $ 39.87 $ & $ 40 $ \\ 
$ 11 $ & $ 11 $ & $ 99.86 $ & $ 100 $ & $ 12 $ & $ 1 $ & $ 40.01 $ & $ 40 $ \\ 
$ 12 $ & $ 11 $ & $ 39.98 $ & $ 40 $ & $ 12 $ & $ 12 $ & $ 99.95 $ & $ 100 $ \\ 
\hline
\end{tabular}
\caption{Comparison between simulated mean of $\Sigma^*_{ij}$ and its theoretical mean}
\label{performance}
\end{table}

\subsubsection{Comparison with the Metropolis-Hastings Algorithm}
\label{metropolis-hastings}
\noindent
A useful Metropolis-Hastings algorithm to sample from the $G$-Wishart distribution has been developed in \cite{mitsakakis2011}.  This approach can also be conveniently adapted to the setting of our generalized $G$-Wishart distribution. Suppose we want to simulate $\Omega$ from a generalized $G$-Wishart distribution with parameters $U$ and $\delta_1, \cdots, \delta_p>0$ corresponding to an undirected graph $G = (V,E)$, with a specified ordering $\sigma$. Let $\Omega = \Phi'\Phi$ and $U^{-1}=T'T$ be the Cholesky decompositions of $\Omega$ and $U^{-1}$. \cite{mitsakakis2011} propose the following algorithm to simulate $\Psi = \Phi T^{-1}$, and thereby $\Omega$ from the required distribution. 

For $i<j$ and $(i,j) \notin E_\sigma$ they call $\Psi_{ij}$ to be a `dependent' entry, while 
$\Psi_I=\{\Psi_{ij}: i=j \mbox{ or } i<j \mbox{ and } (i,j) \in E_\sigma\}$ are referred to as the independent entries of $\Psi$. Also for $1 \leq j \leq p$, they define $\nu_j := |\{s|s>j,\,(s,j) \in E_{\sigma} \}|$. Given the independent entries, the dependent entries can be calculated exactly as in Section \ref{accept-reject}. Let $h(\Psi_I)=\prod_{i=1}^p \chi_{\delta_i+\nu_i} \times N_{|E_\sigma|} \left(0_{|E_\sigma|},I_{|E_\sigma|} \right)$ denote the distribution on $\Psi_I$, where for $i=1,\ldots,p$, $\Psi_{ii}^2 \sim \chi^2_{\delta_i+\nu_i}$ and for $i<j,(i,j) \in E_\sigma$, $\Psi_{ij}$ are independent standard normal. For a given positive integer $N$, the procedure to generate $N$ iterations of the Metropolis-Hastings algorithm is given as follows. 
\begin{itemize}
\item Initialize $\Psi_I^{(0)}$ by sampling $\Psi_I^{(0)}$ from $h(\cdot)$ and set $\Psi_I^{cur} = 
\Psi_I^{(0)}$. 
\item For $i$ in $1,2,\ldots,N$ do ::
\begin{enumerate}
\item Sample $\Psi_I^{prop}$ from $h(\cdot)$.
\item Set $\log \alpha = \frac{1}{2} \sum_{(i,j) \notin E_\sigma} \{(\Psi^{cur}_{ij})^2 - 
(\Psi^{prop}_{ij})^2\}$. 
\item Sample $b$ from Bernoulli$(min(\alpha,1))$.
If $b=1$, set $\Psi_I^{(i)} = \Psi_I^{prop}$, else set $\Psi_I^{(i)} = \Psi_I^{cur}$.
\end{enumerate}
\end{itemize}

\begin{figure}[h]
\begin{minipage}[b]{0.45\linewidth}
\includegraphics[scale=0.45]{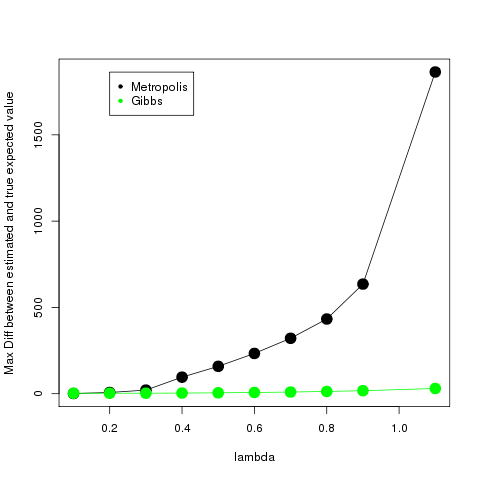}
\end{minipage}
\quad
\begin{minipage}[b]{0.45\linewidth}
\includegraphics[scale=0.45]{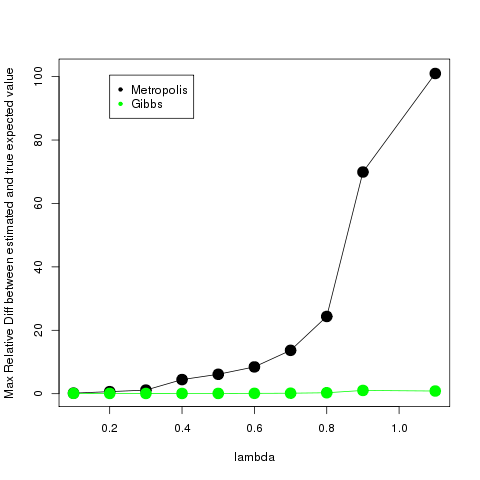}
\end{minipage}
\caption{Plots comparing the maximum entry wise difference and the maximum entry wise relative 
difference between the estimated and true expected values for the Metropolis-Hastings algorithm 
versus the Gibbs sampling algorithm.}
\label{mitsakakis}
\end{figure}

In the above algorithm at each stage the acceptance probability depends on $\alpha$, which then depends on the dependent entries of $\Psi^{cur}$ and $\Psi^{prop}$. The dependent entries in turn depend on the matrix $T$ via terms of the form $\frac{t_{ij}}{t_{jj}}$ for $1 \leq i<j \leq p$. If the terms $\frac{t_{ij}}{t_{jj}}$ are large in magnitude then we expect the terms $\sum_{(i,j) \notin E_\sigma}(\Psi^{cur}_{ij})^2$ and $\sum_{(i,j) \notin E_\sigma} (\Psi^{prop}_{ij})^2$ to be large in magnitude as well. This typically makes $\log \alpha$ either a large positive number or a large negative number which makes the acceptance probability close to $0$ or $1$, thereby making the process potentially expensive in terms of timing. To illustrate this fact lets consider the $5 \times 3$ grid (which satisfies the Generalized Bartlett property) and take all values of $\{t_{ij}:i<j\}$ equal to $\lambda$ for some $\lambda>0$ and the diagonal entries of $T$ as $1$. We illustrate below that as the value of $\lambda$ increases, the performance of the Metropolis-Hastings algorithm can deteriorate. In comparison, this change in $\lambda$ has negligible effect on the performance of the proposed Gibbs sampling algorithm. If $\Omega=\Sigma^{-1}$ is a generalized $G$-Wishart with parameters $U$ and $\boldsymbol{\delta}$, then for $i=j$ or $(i,j) \in E_\sigma$, the  expected value of $\Sigma^*_{ij}$ is $U_{ij}$ (see Theorem \ref{thm2}). Thus, to compare the performance of the two algorithms, we check the difference between the estimated values of $\Sigma^*$ and $U$ (the independent parts only) using the sup norm and also the relative error. Figure \ref{mitsakakis} shows the comparison between the Gibbs algorithm and the MH algorithm (both of which are algorithms for the generalized Wishart class) for varying values of $\lambda$ (choosing the entries of $\boldsymbol{\delta}$ chosen uniformly on a grid from $70$ to $100$ for all cases). The running time for both algorithms for each $\lambda$ is approximately 10 mins on an AMD-V 2.4 GHz processor.

\subsection{Application to temperature data}
\label{climate}

In this section, we provide an illustration of our methods on the (\cite{hadcrut}) dataset. HadCRUT3 dataset consists of monthly temperature data provided on a grid over the globe starting from 1850 till 2012. The spatial resolution is at a $5^{\circ}$ latitude and $5^{\circ}$ longitude. Here, we consider $28$ locations from the US map (out $1732$ locations worldwide) as shown in Figure \ref{usa_nby3}. Thus we have $n=157$ samples for $p=28$ variables. Our goal is to estimate the precision matrix $\Omega$ of these $28$ temperature variables. Given the spatial nature of the data, it is natural to impose sparsity on the precision matrix, with the underlying graph $G$ as shown in Figure \ref{usa_nby3}. We use an ordering $\sigma$ of the variables specified as follows: label the vertex in the bottom of the leftmost column of the grid as $1$, and then move up the columns from south to north, and the rows from east to west. We proceed to fit a concentration graph model, which assumes that $\Omega \in \mathbb{P}_{G_\sigma}$, with $G$ and $\sigma$ defined as above.  

\begin{figure}[h]
\begin{minipage}[b]{0.45\linewidth}
\includegraphics[scale=0.45]{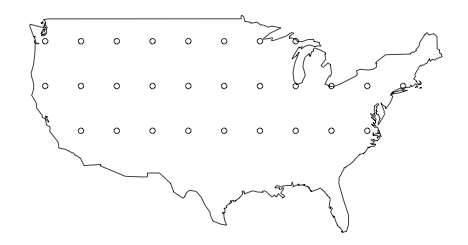}
\end{minipage}
\quad
\begin{minipage}[b]{0.45\linewidth}
\includegraphics[scale=0.45]{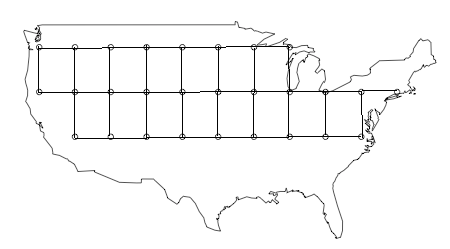}
\end{minipage}
\caption{(left) 28 US grid locations for the HadCRUT3 temperature data. (right) The underlying graph 
$G$ for the concentration graph model.}
\label{usa_nby3}
\end{figure}

The graph $G$ is not decomposable, but can be shown to be a Generalized Bartlett graph (it is an induced subgraph of an $11 \times 3$ grid). Hence, the Bayesian framework developed in this paper can be used to obtain an estimate of $\Omega$. We use two different empirical/objective priors for our analysis. In particular, our first choice is a generalized Wishart prior with scale parameter $U^1 = I_{28}$ and shape parameter ${\bf \delta^1}$ with $\delta^1_i = 1/S_{ii}$ for $1 \leq i \leq 28$. As a second choice, we use a generalized Wishart prior with scale parameter $U^2 = I_{28}$ and shape parameter ${\bf \delta^2}$ with $\delta^2_i = (S^{-1})_{ii}$ for $1 \leq i \leq p$. The Gibbs sampling procedure specified in Section 5.2 was used to generate samples from the two corresponding posterior distributions. The burn-in period was chosen to be $2000$ iterations, and the subsequent $1000$ iterations were used to compute the posterior means and credible intervals. Increasing the burn-in to more than $2000$ iterations did not lead to significant changes in the estimates, thus indicating that the chosen burn-in period is appropriate. The posterior mean estimates for both the priors are provided in Table \ref{paleo1}. The MLE for $\Omega \in \mathbb{P}_{G_\sigma}$ was also computed using the glasso function in $R$, and is provided in Table \ref{paleo1} as well. As noted in the introduction, an inherent advantage of Bayesian methods is ability to easily provide uncertainty quantification using the posterior distribution. The estimated $95 \%$ posterior credible intervals for both prior choices are provided in Table \ref{paleo2}.

{\scriptsize
\begin{table}
\scriptsize
\centering
\caption{Posterior expectations for the two priors and mle estimate of the precision matrix for the temperature data in Section \ref{climate}.} 
\begin{tabular}{ccc}
\begin{tabular}{|r|r|r|r|}
\hline
(i,j) & Bayes 1 & Bayes 2 & Glasso\\ \hline
(2,1) & -2.52 & -2.76 & -2.59 \\ 
(4,1) & -1.51 & -1.64 & -1.52 \\ 
(4,3) & -1.20 & -1.25 & -1.17 \\ 
(5,2) & -2.75 & -2.86 & -2.81 \\ 
(5,4) & -1.32 & -1.40 & -1.35 \\ 
(6,3) & -2.37 & -2.40 & -2.38 \\ 
(7,4) & -3.37 & -3.57 & -3.53 \\ 
(7,6) & -0.72 & -0.82 & -0.71 \\ 
(8,5) & -3.04 & -3.20 & -3.08 \\ 
(8,7) & 1.12 & 1.21 & 1.18 \\ 
(9,6) & -3.58 & -4.01 & -3.73 \\ 
(10,7) & -1.23 & -1.34 & -1.26 \\ 
(10,9) & -1.44 & -1.39 & -1.37 \\ 
(11,8) & -5.38 & -6.26 & -5.82 \\ 
(11,10) & -1.52 & -1.66 & -1.54 \\ 
(12,9) & -1.13 & -1.10 & -1.05 \\ 
(13,10) & -2.21 & -2.44 & -2.29 \\ 
(13,12) & -2.62 & -2.83 & -2.65 \\ 
(14,11) & -0.49 & -0.55 & -0.50 \\ 
(14,13) & -0.38 & -0.43 & -0.37 \\ 
(15,12) & -3.19 & -3.44 & -3.25 \\ 
(16,13) & -0.48 & -0.54 & -0.45 \\ 
(16,15) & -0.67 & -0.75 & -0.67 \\ 
(17,14) & -0.42 & -0.40 & -0.40 \\ 
\hline
\end{tabular} & 
\begin{tabular}{|r|r|r|r|}
\hline
(i,j) & Bayes 1 & Bayes 2 & Glasso\\ \hline
(17,16) & -0.10 & -0.09 & -0.07 \\ 
(18,15) & -2.48 & -2.83 & -2.57 \\ 
(19,16) & -0.68 & -0.71 & -0.66 \\ 
(19,18) & -1.46 & -1.50 & -1.43 \\ 
(20,17) & -3.19 & -3.50 & -3.36 \\ 
(20,19) & -2.45 & -2.89 & -2.54 \\ 
(21,18) & -2.87 & -3.01 & -2.96 \\ 
(22,19) & -1.88 & -2.10 & -1.83 \\ 
(22,21) & -2.37 & -2.63 & -2.40 \\ 
(23,20) & -3.27 & -3.62 & -3.40 \\ 
(23,22) & -2.18 & -2.29 & -2.20 \\ 
(24,21) & -0.84 & -0.96 & -0.84 \\ 
(25,22) & -5.16 & -5.34 & -5.36 \\ 
(25,24) & -0.71 & -0.66 & -0.70 \\ 
(26,24) & -0.68 & -0.69 & -0.66 \\ 
(27,25) & -7.53 & -9.25 & -8.19 \\ 
(27,26) & -1.65 & -1.75 & -1.68 \\ 
(28,27) & -3.77 & -4.25 & -3.87 \\
(1,1) & 4.62 & 5.04 & 4.65 \\ 
(2,2) & 7.26 & 7.65 & 7.39 \\ 
(3,3) & 4.88 & 4.99 & 4.80 \\ 
(4,4) & 6.79 & 7.22 & 6.98 \\ 
(5,5) & 6.77 & 7.12 & 6.84 \\ 
\hline
\end{tabular} & 
\begin{tabular}{|r|r|r|r|}
\hline
(i,j) & Bayes 1 & Bayes 2 & Glasso\\ \hline
(6,6) & 8.40 & 9.25 & 8.53 \\ 
(7,7) & 4.15 & 4.46 & 4.26 \\ 
(8,8) & 8.65 & 9.90 & 9.20 \\ 
(9,9) & 6.41 & 6.54 & 6.26 \\ 
(10,10) & 6.15 & 6.60 & 6.18 \\ 
(11,11) & 6.54 & 7.47 & 6.93 \\ 
(12,12) & 8.01 & 8.52 & 8.02 \\ 
(13,13) & 6.31 & 7.04 & 6.33 \\ 
(14,14) & 0.77 & 0.79 & 0.75 \\ 
(15,15) & 6.74 & 7.55 & 6.89 \\ 
(16,16) & 1.23 & 1.32 & 1.18 \\ 
(17,17) & 3.53 & 3.79 & 3.61 \\ 
(18,18) & 6.05 & 6.45 & 6.14 \\ 
(19,19) & 6.54 & 7.31 & 6.52 \\ 
(20,20) & 9.91 & 11.10 & 10.34 \\ 
(21,21) & 7.23 & 7.89 & 7.36 \\ 
(22,22) & 10.81 & 11.47 & 10.99 \\ 
(23,23) & 5.03 & 5.46 & 5.14 \\ 
(24,24) & 2.31 & 2.36 & 2.27 \\ 
(25,25) & 14.69 & 16.93 & 15.62 \\ 
(26,26) & 4.71 & 4.99 & 4.72 \\ 
(27,27) & 11.53 & 13.52 & 12.20 \\ 
(28,28) & 5.46 & 6.04 & 5.55 \\ \hline
\end{tabular}
\end{tabular}
\label{paleo1}
\end{table}

\begin{table}
\scriptsize
\centering
\caption{$95\%$ credible interval for elements of the precision matrix corresponding to the two priors for the temperature data in Section \ref{climate}.}
\begin{tabular}{ccc}
\begin{tabular}{|r|r|r|}
\hline 
 & \multicolumn{2}{c|}{95\% CI for the posterior mean}\\ \hline
(i,j) & Bayes 1 & Bayes 2\\ \hline
(2,1) & (-2.59,-2.46) & (-2.83,-2.69)  \\ 
(4,1) & (-1.56,-1.47) & (-1.69,-1.59)  \\ 
(4,3) & (-1.24,-1.16) & (-1.29,-1.20)  \\ 
(5,2) & (-2.82,-2.68) & (-2.92,-2.79)  \\ 
(5,4) & (-1.36,-1.28) & (-1.44,-1.36)  \\ 
(6,3) & (-2.43,-2.30) & (-2.46,-2.33)  \\ 
(7,4) & (-3.43,-3.31) & (-3.64,-3.50)  \\ 
(7,6) & (-0.76,-0.68) & (-0.86,-0.78)  \\ 
(8,5) & (-3.10,-2.98) & (-3.26,-3.15)  \\ 
(8,7) & (1.10,1.15) & (1.18,1.24)  \\ 
(9,6) & (-3.67,-3.50) & (-4.09,-3.92)  \\ 
(10,7) & (-1.27,-1.20) & (-1.37,-1.30)  \\ 
(10,9) & (-1.49,-1.39) & (-1.44,-1.33)  \\ 
(11,8) & (-5.47,-5.28) & (-6.37,-6.16)  \\ 
(11,10) & (-1.56,-1.48) & (-1.69,-1.62)  \\ 
(12,9) & (-1.18,-1.08) & (-1.15,-1.05)  \\ 
(13,10) & (-2.27,-2.16) & (-2.50,-2.38)  \\ 
(13,12) & (-2.69,-2.55) & (-2.90,-2.75)  \\ 
(14,11) & (-0.51,-0.48) & (-0.57,-0.54)  \\ 
(14,13) & (-0.40,-0.37) & (-0.45,-0.41)  \\ 
(15,12) & (-3.26,-3.12) & (-3.52,-3.37)  \\ 
(16,13) & (-0.50,-0.45) & (-0.57,-0.52)  \\ 
(16,15) & (-0.70,-0.64) & (-0.78,-0.72)  \\ 
(17,14) & (-0.43,-0.40) & (-0.41,-0.39)  \\ 
 \hline
\end{tabular} &
\begin{tabular}{|r|r|r|}
\hline 
 & \multicolumn{2}{c|}{95\% CI for the posterior mean}\\ \hline
(i,j) & Bayes 1 & Bayes 2\\ \hline
(17,16) & (-0.11,-0.08) & (-0.11,-0.08)  \\ 
(18,15) & (-2.54,-2.42) & (-2.89,-2.77)  \\ 
(19,16) & (-0.70,-0.65) & (-0.74,-0.68)  \\ 
(19,18) & (-1.51,-1.41) & (-1.55,-1.45)  \\ 
(20,17) & (-3.26,-3.12) & (-3.56,-3.43)  \\ 
(20,19) & (-2.51,-2.38) & (-2.96,-2.82)  \\ 
(21,18) & (-2.94,-2.80) & (-3.08,-2.94)  \\ 
(22,19) & (-1.95,-1.82) & (-2.17,-2.03)  \\ 
(22,21) & (-2.44,-2.30) & (-2.70,-2.55)  \\ 
(23,20) & (-3.34,-3.20) & (-3.69,-3.54)  \\ 
(23,22) & (-2.24,-2.12) & (-2.35,-2.23)  \\ 
(24,21) & (-0.88,-0.81) & (-0.99,-0.92)  \\ 
(25,22) & (-5.26,-5.06) & (-5.45,-5.23)  \\ 
(25,24) & (-0.75,-0.67) & (-0.71,-0.62)  \\ 
(26,24) & (-0.72,-0.64) & (-0.73,-0.65)  \\ 
(27,25) & (-7.67,-7.39) & (-9.41,-9.09)  \\ 
(27,26) & (-1.70,-1.60) & (-1.81,-1.70)  \\ 
(28,27) & (-3.85,-3.70) & (-4.33,-4.16)  \\ 
(1,1) & (4.54,4.70) & (4.97,5.12)  \\ 
(2,2) & (7.14,7.38) & (7.53,7.77)  \\ 
(3,3) & (4.81,4.96) & (4.92,5.05)  \\ 
(4,4) & (6.70,6.88) & (7.11,7.33)  \\ 
(5,5) & (6.67,6.87) & (7.02,7.21)  \\ 
 \hline
\end{tabular} & 
\begin{tabular}{|r|r|r|}
\hline 
 & \multicolumn{2}{c|}{95\% CI for the posterior mean}\\ \hline
(i,j) & Bayes 1 & Bayes 2\\ \hline
(6,6) & (8.26,8.54) & (9.12,9.38)  \\ 
(7,7) & (4.09,4.21) & (4.39,4.53)  \\ 
(8,8) & (8.52,8.78) & (9.76,10.04)  \\ 
(9,9) & (6.31,6.51) & (6.43,6.64)  \\ 
(10,10) & (6.05,6.24) & (6.51,6.69)  \\ 
(11,11) & (6.44,6.64) & (7.37,7.57)  \\ 
(12,12) & (7.89,8.14) & (8.40,8.65)  \\ 
(13,13) & (6.22,6.40) & (6.94,7.13)  \\ 
(14,14) & (0.75,0.78) & (0.78,0.80)  \\ 
(15,15) & (6.64,6.84) & (7.45,7.65)  \\ 
(16,16) & (1.21,1.25) & (1.30,1.34)  \\ 
(17,17) & (3.47,3.59) & (3.73,3.84)  \\ 
(18,18) & (5.96,6.13) & (6.37,6.53)  \\ 
(19,19) & (6.45,6.64) & (7.21,7.41)  \\ 
(20,20) & (9.78,10.05) & (10.96,11.24)  \\ 
(21,21) & (7.11,7.36) & (7.76,8.02)  \\ 
(22,22) & (10.67,10.96) & (11.32,11.62)  \\ 
(23,23) & (4.95,5.12) & (5.38,5.55)  \\ 
(24,24) & (2.27,2.35) & (2.32,2.40)  \\ 
(25,25) & (14.48,14.89) & (16.71,17.15)  \\ 
(26,26) & (4.63,4.80) & (4.90,5.07)  \\ 
(27,27) & (11.37,11.70) & (13.34,13.70)  \\ 
(28,28) & (5.37,5.56) & (5.94,6.14)  \\ \hline
\end{tabular}
\end{tabular}
\label{paleo2}
\end{table}}

\newpage
\newpage

\pagenumbering{gobble}
\appendix

\bibliographystyle{rss}
\bibliography{nondecom_bibliography}

\newpage
\pagenumbering{arabic}
\section*{Supplemental Document }
\label{suppsec}
\section{Proofs}
\subsection{Proof of Theorem \ref{thm1}}
\label{proofthm1}

Note that,
\begin{eqnarray*}
tr \left(LDL^T U \right)&=& tr \left(DL^TUL \right) \\
&=& \sum_{j=1}^p D_j L_{.j}^T U L_{.j} \mbox{ [$L_{.j}$ is the $j$-th column of $L$]} \\
\end{eqnarray*}
where $L_{.j}$ is the $j$-th column of $L$. Since $U$ is positive definite, let $\lambda>0$ be its minimum eigenvalue. Hence,
\begin{eqnarray*}
tr \left(LDL^T U \right) &\geq & \lambda \sum_{j=1}^p D_j L_{.j}^T  L_{.j} \geq  \lambda \sum_{j=1}^p \left(D_j + D_j L_{I_j}^T  L_{I_j} \right)	 \\
\end{eqnarray*}
where $L_{I_j}$ denote the vector of independent entries of $L_{.j}$ and the length of $L_{I_j}$ is $\nu_j$.
Hence, there exists $c_1>0$ such that,
\[\pi^*_{U,\boldsymbol{\delta}}(L_I,D) \leq c_1 \prod_{j=1}^p D_j^{\frac{\delta_j}{2}+\nu_j} \exp \left(-\frac{\lambda D_j}{2}\right) \exp \left(-\frac{\lambda D_j}{2} L_{I_j}^T  L_{I_j}\right).\]
Now, 
\[\int_{L_{I_j} \in \mathbb{R}^{\nu_j}} \exp \left(-\frac{\lambda D_j}{2} L_{I_j}^T  L_{I_j} \right) dL_{I_j}=(2\pi)^{\frac{\nu_j}{2}} \lambda^{-\frac{\nu_j}{2}} D_j^{-\frac{\nu_j}{2}}\]
and we know that
\begin{eqnarray*}
&& \int D_j^{\frac{\delta_j}{2}+\nu_j-\frac{\nu_j}{2}} \exp \left(-\frac{\lambda D_j}{2}\right) dD_j \\
&=& \int D_j^{\frac{\delta_j+\nu_j}{2}} \exp \left(-\frac{\lambda D_j}{2}\right) dD_j \\
&<& \infty,\\
\end{eqnarray*}
since $\delta_j>0$. Hence $\pi^*_{U,\boldsymbol{\delta}}$ can be normalized to a proper density $\pi_{U,\boldsymbol{\delta}}$. Next we prove that $\forall i \geq j$, $\Omega_{ij}$ has finite expectation under this density. Since $\Omega_{ij}=\sum_{k \leq j} L_{ik}L_{jk}D_k$, it is enough to show that $\forall k \leq j$, the expectation of $L_{ik}L_{jk}D_k$ exists. Let us consider the following cases,
\begin{description}
\item[Case 1] $i>j>k$\\
\[L_{ik}L_{jk}D_k = (L_{ik}\sqrt{D_k})(L_{jk}\sqrt{D_k})\]
\item[Case 2] $i=j>k$\\
\[L_{ik}L_{jk}D_k = (|L_{ik}|\sqrt{D_k})^2\]
\item[Case 3] $i=j=k$\\
\[L_{ik}L_{jk}D_k = D_k\]
\item[Case 4] $i>j=k$\\
\[L_{ik}L_{jk}D_k = \sqrt{D_k}(L_{ik}\sqrt{D_k})\]
\end{description}
It follows from equation \ref{defpistar}, that
\begin{eqnarray*}
|L_{ik}L_{jk}D_k|\pi^*_{U,\boldsymbol{\delta}}(L_I,D) &\leq & c_1 |L_{ik}L_{jk}D_k| \prod_{l=1}^p D_l^{\frac{\delta_l}{2}+\nu_l} \exp \left(-\frac{\lambda D_l}{2}\right) \exp \left(-\frac{\lambda D_l}{2} L_{.l}^T  L_{.l}\right)\\
&=& |L_{ik}L_{jk}D_k| \prod_{l=1}^p  \exp \left(-\frac{\lambda D_l}{4}\right) \exp \left(-\frac{\lambda D_l}{4} L_{.l}^T  L_{.l}\right)\\
&\times & c_1 \prod_{l=1}^p D_l^{\frac{\delta_l}{2}+\nu_l} \exp \left(-\frac{\lambda D_l}{4}\right) \exp \left(-\frac{\lambda D_l}{4} L_{.l}^T  L_{.l}\right) \hspace{2cm} (*)
\end{eqnarray*}
Note that both $x^2\exp \left(- \frac{\lambda x^2}{4} \right)$ and $|x|\exp \left(- \frac{\lambda x^2}{4} \right)$ are uniformly bounded above in $x$. It follows by $(*)$ that in all the cases considered above, $|L_{ik}L_{jk} D_k| \exp \left(-\frac{\lambda D_k}{4}\right) \exp \left(-\frac{\lambda D_k}{4} L_{.k}^T  L_{.k}\right)$ is uniformly bounded in $(L_I,D)$. Since we have already established that 
\[\prod_{j=1}^p D_j^{\frac{\delta_j}{2}+\nu_j} \exp \left(-\frac{\lambda D_j}{4}\right) \exp \left(-\frac{\lambda D_j}{4} L_{.j}^T  L_{.j}\right)\] is integrable, it follows that $L_{ik}L_{jk}D_k$ has finite expectation under $\pi_{U,\boldsymbol{\delta}}$.

\subsection{Proof of Theorem \ref{thm2}}
\label{proofthm2}
Let 
\[z_G(U,\boldsymbol{\delta}) := \int \exp \left(-\frac{1}{2}tr(\Omega U) + \sum_{k \leq p} \frac{\delta_k}{2} \log(D_k) \right) d\Omega\]

Let $\Omega_k$ denote the principal submatrix corresponding to the first $k$ rows and columns of $\Omega$. Then $D_k = \frac{|\Omega_k|}{|\Omega_{k-1}|}$. Thus 
\[\sum_{k \leq p} \frac{\delta_k}{2} \log(D_k) = \sum_{k \leq p} \log(|\Omega_k|) \left(\frac{\delta_k-\delta_{k+1}}{2} \right) \]
where $\delta_{p+1}=0$. Hence,
\[\int \frac{1}{z_G(U,\boldsymbol{\delta})} \exp \left(-\frac{1}{2}tr(\Omega U) + \sum_{k \leq p} \log(|\Omega_k|)\left(\frac{\delta_k-\delta_{k+1}}{2} \right) \right) d\Omega = 1\]
Differentiating both sides w.r.t. $\Omega_{ij}$ and assuming that we can take the derivative inside the integral, we get
\[E\left(-U_{ij} + \sum_{k \leq p} \frac{\partial \log(|\Omega_k|)}{\partial \Omega_{ij}} \left(\frac{\delta_k-\delta_{k+1}}{2} \right) \right) = 0 \mbox{ \hspace{5cm}}(\bigstar)\]
Since $\frac{\partial \log(|\Omega_k|)}{\partial \Omega_{ij}} = 2(\Omega_k^{-1})_{ij}$ for $k \geq \max(i,j)$ and $0$ otherwise, we observe that
\[U_{ij} = E \left( \sum_{\max(i,j) \leq k \leq p} (\Omega_k^{-1})_{ij} (\delta_k-\delta_{k+1}) \right)\]

\noindent 
We now rigorously establish the validity of exchanging the derivative and the integral mentioned above.
Define, 
\[I_{G_\sigma} = \{\Omega \in \mathbb{M}_p| \Omega_{ii} \geq 0, 1\leq i \leq p, \mbox{ and } \Omega_{ij}=0 \mbox{ if } (i,j) \notin \widetilde{E}\} \]
\noindent
We see that $\mathbb{P}_{G_\sigma}$ is an open subset of $I_{G_\sigma}$ and has a positive measure (with respect to induced Lebesgue measure on $I_{G_\sigma}$). Thus $\pi_{U,\boldsymbol{\delta}}()$ can also be seen as a density on $I_{G_\sigma}$ taking value $0$ on $I_{G_\sigma}-\mathbb{P}_{G_\sigma}$.
\noindent
We claim that $\pi_{U,\boldsymbol{\delta}}()$ is continuous on $I_{G_\sigma}$. Clearly, it is enough to prove this claim on the boundary $\mathbb{B}_{G_\sigma}$ between $\mathbb{P}_{G_\sigma}$ and $I_{G_\sigma}-\mathbb{P}_{G_\sigma}$. Since $\mathbb{P}_{G_\sigma}$ is open, $\mathbb{B}_{G_\sigma} \subset I_{G_\sigma}-\mathbb{P}_{G_\sigma}$. Thus for $\Omega \in \mathbb{B}_{G_\sigma}$, $\pi_{U,\boldsymbol{\delta}}(\Omega)=0$. 

Let $\Omega \in \mathbb{B}_{G_\sigma}$ and $\{\Omega^{(n)} \}$ be a sequence in $\mathbb{P}_{G_\sigma}$ such that $\Omega^{(n)} \rightarrow \Omega$. We want to show $\pi_{U,\boldsymbol{\delta}}(\Omega^{(n)}) \rightarrow \pi_{U,\boldsymbol{\delta}}(\Omega)=0$. If, $\Omega^{(n)}_k$ denote the submatrix of $\Omega^{(n)}$ corresponding to the first $k$ rows and columns, then
\[\Omega^{(n)}_{k+1} = \left( \begin{array}{cc}
\Omega^{(n)}_k & (\Omega^{(n)}_{k+1,1:k})^T \\
\Omega^{(n)}_{k+1,1:k} & \Omega^{(n)}_{k+1,k+1} \\
\end{array} \right)\] 
and \[\frac{|\Omega^{(n)}_{k+1}|}{|\Omega^{(n)}_{k}|} = \left( \Omega^{(n)}_{k+1,k+1} - (\Omega^{(n)}_{k+1,1:k})^T (\Omega^{(n)}_k)^{-1} \Omega^{(n)}_{k+1,1:k} \right) .\]
\noindent
Since $\Omega^{(n)}$ is positive definite, the above ratio is postie and less than equal to $\Omega^{(n)}_{k+1,k+1}$. But $\Omega^{(n)}_{k+1,k+1} \rightarrow \Omega_{k+1,k+1}$, meaning that the sequence $\{|\Omega^{(n)}_{k+1}|/|\Omega^{(n)}_{k}|\}$  is positive and bounded above. Since $\Omega$ is not positive definite, either $\Omega_{11}=0$, or there exists $k_0$ such that $|\Omega_{k_0+1,k_0+1}| = 0$ and $|\Omega_{k_0,k_0}| > 0$.

The exponential term in $\pi_{U,\boldsymbol{\delta}}(\Omega^{(n)})$  is $\exp \left(-tr(\Omega^{(n)}U)/2 \right) < 1$ and in both cases mentioned above we have atleast one term outside the exponential converging to $0$, while the rest are bounded above; which proves that  $\pi_{U,\boldsymbol{\delta}}(\Omega^{(n)}) \rightarrow 0$.

Next for $\Omega \in \mathbb{P}_{G_\sigma}$ we find the partial derivatives and double derivatives of $\pi_{U,\boldsymbol{\delta}}(\Omega)$. For notational convenience we replace $\pi_{U,\boldsymbol{\delta}}$ by $\pi()$ for the rest of this proof. We have already seen that for $\Omega \in \mathbb{P}_{G_\sigma}$,
\[\frac{\partial \pi(\Omega)}{\partial \Omega_{ij}} = \pi(\Omega) \left[-U_{ij} + \sum_{1 \leq k \leq p} (\Omega_k^{-1})_{ij} (\delta_k-\delta_{k+1})\right]\]
Note that for $k < \max(i,j)$, $(\Omega_k^{-1})_{ij}=0$. Differentiating the above equation again we get,
\[\frac{\partial^2 \pi(\Omega)}{\partial \Omega_{ij}^2} = \pi(\Omega) \left[-U_{ij} +\sum_{1 \leq k \leq p} (\Omega_k^{-1})_{ij} (\delta_k-\delta_{k+1})\right]^2 + \pi(\Omega) \left[\sum_{1 \leq k \leq p}  (\delta_k-\delta_{k+1}) \frac{\partial (\Omega_k^{-1})_{ij}}{\partial \Omega_{ij}}\right] \ldots (*)\]
\noindent
For a symmetric non-singular matrix $A$,
\begin{eqnarray*}
\frac{\partial (A^{-1})_{ij}}{\partial A_{ij}} &=& \frac{\partial}{\partial A_{ij}} \left[\frac{|A^{(ij)}|}{|A|} \right] \\
&=& \left[ -\frac{2|A^{(ij)}|^2}{|A|^2} + \frac{1}{|A|} \frac{\partial |A^{(ij)}|}{\partial A_{ij}}\right] \hspace{2cm} \ldots (**)
\end{eqnarray*}

\noindent
For $\Omega$ in the interior of $I_{G_\sigma}-\mathbb{P}_{G_\sigma}$, above partial derivatives and double derivatives exists and equals $0$. Now, fix $i>j$ and consider $\Omega$ in $\mathbb{B}_{G_\sigma}$. Let us define $\Omega^{h}$ to be a $p \times p$ matrix such that, $\Omega^{h}_{rs} = \Omega_{rs}$ if $(r,s) \neq (i,j)$ and $\Omega^{h}_{ij} = \Omega_{ij}+h$.
If $\Omega^{h} \in I_{G_\sigma}-\mathbb{P}_{G_\sigma}$ then $\frac{\pi(\Omega^{h}) - \pi(\Omega)}{h} = 0$. We now show that for $h_n \rightarrow 0$, such that  $\Omega^{h_n} \in \mathbb{P}_{G_\sigma},\forall n$
\[\frac{\pi(\Omega^{h_n}) - \pi(\Omega)}{h_n} = \frac{\pi(\Omega^{h_n})}{h_n} \rightarrow 0.\] 

\noindent
This will show that for $\Omega \in \mathbb{B}_{G_\sigma}$, $\frac{\partial \pi(\Omega)}{\partial \Omega_{ij}}$ exist and equal $0$. Since $\Omega^{h_n}$ is positive definite $\Omega_{11}>0$. Hence $\exists \, k_0$ such that $|\Omega_{k_0-1}|>0$ and $|\Omega_{k_0}|=0$. Consider $\frac{|\Omega^{h_n}_{k_0}|}{|\Omega^{h_n}_{k_0-1}|}$. Note that $|\Omega^{h_n}_{k_0-1}| \rightarrow |\Omega_{k_0-1}|>0$, $|\Omega^{h_n}_{k_0}| \rightarrow |\Omega_{k_0}|=0$ and $|\Omega^{h}_{k_0}|$ is a quadratic in $h$, which equals $0$ at $h=0$. Hence $|\Omega^{h}_{k_0}|$ can be written as $ah^2+bh$. If $\delta_{k_0}>2$, then
\begin{eqnarray*}
\frac{1}{h_n} \left(\frac{|\Omega^{h_n}_{k_0}|}{|\Omega^{h_n}_{k_0-1}|} \right)^{(\delta_{k_0}/2)} &=& 
\frac{1}{|\Omega^{h_n}_{k_0-1}|^{(\delta_{k_0}/2)}} \frac{(ah_n^2+bh_n)^{(\delta_{k_0}/2)}}{h_n} \\
&=& \frac{1}{|\Omega^{h_n}_{k_0-1}|^{(\delta_{k_0}/2)}} h_n^{(\delta_{k_0}/2-1)} (ah_n+b)^{(\delta_{k_0}/2)}\\
&\rightarrow & 0 \\  
\end{eqnarray*}
\noindent
Now note that the other terms in $\frac{\pi(\Omega^{h_n})}{h_n}$ are either  $\left( \frac{|\Omega^{h_n}_{k}|}{|\Omega^{h_n}_{k-1}|} \right)^{(\delta_k/2)}$ or $\exp(-tr(\Omega^{h_n}U)/2)$; which are all positive and bounded above. Thus if $\delta_k>2 ,\, \forall k$ the partial derivatives exist and equal $0$ at $\mathbb{B}_{G_\sigma}$. Before exploring the partial double derivatives we state the following fact, which is straightforward to establish.
\begin{fact}
If $p(\Omega)$ is a polynomial in elements of $\Omega \in \mathbb{P}_{G_\sigma}$ and $U$ is p.d. then $|p(\Omega)|\exp(-tr(U\Omega))$
as a function of $\Omega$ is uniformly bounded above.
\end{fact}
\noindent
If $\delta_k >4,\forall k$ the existence of partial double derivatives on $\mathbb{B}_{G_\sigma}$ can proved with the the help of the above fact and  similar arguments as in the case of single partial derivatives. 

If $\pi_{ij}(\Omega_{ij})$ is the marginal density of $\Omega_{ij}$ then,
\[\pi_{ij}(\Omega_{ij}) = \int \pi(\Omega) d (\Omega\setminus \Omega_{ij})\]

Note that $\pi_{ij}$ has support over whole real line. For arbitrary $\Omega_{ij}^0 \in \mathbb{R}$,
\begin{eqnarray*}
&& \left|\frac{\pi_{ij}(\Omega_{ij}^0+h)-\pi_{ij}(\Omega_{ij}^0)}{h} - \int \frac{\partial \pi(\Omega)}{\partial \Omega_{ij}} d (\Omega\setminus \Omega_{ij})\right| \\
&\leq & \int \left| \frac{\pi(\Omega\setminus \Omega_{ij},\Omega_{ij}^0+h)-\pi(\Omega\setminus \Omega_{ij},\Omega_{ij}^0)}{h} - \frac{\partial \pi(\Omega)}{\partial \Omega_{ij}}|_{\Omega_{ij}=\Omega_{ij}^0} \right| d (\Omega\setminus \Omega_{ij}) \\
& \leq & h \int \left| \frac{\partial^2 \pi(\Omega)}{\partial \Omega_{ij}^2} |_{\Omega_{ij}=\Omega_{ij}^0+z}\right|  d (\Omega\setminus \Omega_{ij}) \mbox{ [for some $z \in [0,h]$]}\\
&\leq & h \int \sup_{\{\Omega|\Omega_{ij} \in (\Omega_{ij}^0,\Omega_{ij}^0+h)\}} \left| \frac{\partial^2 \pi(\Omega)}{\partial \Omega_{ij}^2} \right|  d (\Omega\setminus \Omega_{ij}) \\
&=& h \int \sup_{\{\Omega \in \mathbb{P}_{G_\sigma}|\Omega_{ij} \in (\Omega_{ij}^0,\Omega_{ij}^0+h)\}} \left| \frac{\partial^2 \pi(\Omega)}{\partial \Omega_{ij}^2} \right|  d (\Omega\setminus \Omega_{ij}) \hspace{2cm} \ldots (***)\\
\end{eqnarray*}
\noindent
If,
\begin{eqnarray*} 
f_1(\Omega) &=& \exp(-tr(\Omega U)/4) \prod_{j=1}^p |\Omega_k|^{(\delta_k-\delta_{k+1})/2} \times \\
&& \left[ \left(-U_{ij} + \sum_{k=1}^p(\delta_k-\delta_{k+1})(\Omega_k^{-1})_{ij} \right)^2 + \sum_{k=1}^p (\delta_k-\delta_{k+1}) \left( -2\frac{|\Omega_k^(ij)|^2}{|\Omega_k|^2} + \frac{1}{|\Omega_k|} \frac{\partial |\Omega_k^(ij)|}{\partial \Omega_{ij}}\right) \right]
\end{eqnarray*}
and $f_2(\Omega)=\exp(-tr(\Omega U)/4)$ then $\frac{\partial^2 \pi(\Omega)}{\partial \Omega_{ij}^2} = f_1(\Omega) f_2(\Omega)$. If $\delta_k>4, \forall k$ then by the fact above, $f_1(\Omega)$ is uniformly bounded above on $\mathbb{P}_{G_\sigma}$ by a constant $M>0$. 
Thus,
\begin{eqnarray*}
&&\int \sup_{\{\Omega \in \mathbb{P}_{G_\sigma}|\Omega_{ij} \in (\Omega_{ij}^0,\Omega_{ij}^0+h)\}} \left| \frac{\partial^2 \pi(\Omega)}{\partial \Omega_{ij}^2} \right|  d (\Omega\setminus \Omega_{ij}) \\
& \leq &\int \sup_{\{\Omega \in \mathbb{P}_{G_\sigma}|\Omega_{ij} \in (\Omega_{ij}^0,\Omega_{ij}^0+h)\}}   M \exp(-tr(\Omega U)/4)d (\Omega\setminus \Omega_{ij}) \\
\end{eqnarray*}

The above integral is finite, which means we can take limit as $h \rightarrow 0$, in $(***)$ to obtain,
\[\frac{d \pi_{ij}(\Omega_{ij})}{d \Omega_{ij}} = \int \frac{\partial \pi(\Omega)}{\partial \Omega_{ij}} d (\Omega\setminus \Omega_{ij})\]
Since the support of $\pi_{ij}()$ is the entire real line, $\int \frac{d \pi_{ij}(\Omega_{ij})}{d \Omega_{ij}}d \Omega_{ij} = 0$. Thus $\int \frac{\partial \pi(\Omega)}{\partial \Omega_{ij}} d\Omega = 0$ which establishes $\bigstar$.

\subsection{Proof of Theorem \ref{thmdecomposableclosedform}}
\label{proofthmdecomposableclosedform}
If $L_I$ are the independent entries in $L$ then the Jacobian of the transformation from $\Omega \rightarrow (L_I,D)$, is $\prod_{j=1}^p D_j^{\nu_j}$.  Thus the unnormalized density of $(L_I,D)$ is,
 \[\pi^*_{U,\boldsymbol{\delta}}(L_I,D) = \prod_{j=1}^p D_j^{\frac{\delta_j}{2}+\nu_j} \exp \left(-\frac{1}{2} tr(LDL^T U) \right)\]
Since $G_\sigma$ is decomposable, and the vertices have been ordered by a perfect elimination ordering,
\[\Omega \in \mathbb{P}_{G_\sigma} \Leftrightarrow L \in \mathcal{L}_{G_\sigma}\] 
Now,
\begin{eqnarray*}
tr(LDL^TU) &=& tr(DL^TUL) = \sum_{j=1}^p D_{j} (L^TUL)_{jj} = \sum_{j=1}^{p-1} D_{j} L_{.j}^TUL_{.j} + D_pU_{pp} \\
\end{eqnarray*}
where $L_{.j}$ is the $j$-th column of $L$. Now $L_{ij}=0$ for $i<j$, $L_{jj}=1$ and since $G_{\sigma}$ is decomposable, for $i>j$, $L_{ij}=0$ if $(i,j) \notin E_{\sigma}$. Thus, if $L_{I_j}$ denotes the set of independent entries in the $j$-th column of $L$, then,
\begin{eqnarray*}
L_{.j}^TUL_{.j} &=& \left(1,L_{I_j}^T \right) \left( \begin{array}{cc}
U_{jj} & \bm{U^{\succ T}_{.j}} \\
\bm{U^{\succ}_{.j}} & U^{\succ j} \\
\end{array} \right)  \left(\begin{array}{c}
1 \\
L_{I_j}\\
\end{array} \right) \\
&=& (L_{I_j}-\bm{e_j})^T U^{\succ j} (L_{I_j}-\bm{e_j}) + c_j\\
\end{eqnarray*}

Thus,
\begin{eqnarray*}
\pi^*_{U,\boldsymbol{\delta}}(L_I,D) &=& \prod_{j=1}^{p-1}\left[ D_j^{\frac{\delta_j}{2}+\nu_j} \exp \left( -\frac{c_j D_j}{2} \right) \exp \left(-\frac{D_j}{2} (L_{I_j}-\bm{e_j})^T U^{\succ j} (L_{I_j}-\bm{e_j}) \right) \right] \\
&& \times D_p^{\frac{\delta_p}{2}} \exp \left(-\frac{c_p D_p}{2} \right) 
\end{eqnarray*}
which implies that $\{ (L_{I_j},D_j)_{j=1}^{p-1},D_p \}$ are mutually independent. After straightforward calculations we find that,

\begin{eqnarray*}
L_{I_j}|D_j &\sim & N\left( \bm{e_j},\frac{(U^{\succ j})^{-1}}{D_j} \right) \\
&and &\\
D_j &\sim & Gamma \left( \frac{\nu_j+\delta_j}{2}+1,\frac{c_j}{2} \right)\\
\end{eqnarray*}

\noindent
Now for $i \geq j$, we calculate the expectation of $\Omega_{ij} = \sum_{k \leq j} L_{ik} L_{jk} D_k$. Note that,
\[E(D_k) = \frac{\delta_k + \nu_k +2}{c_k}, \]
and for $i \geq j >k$,
\begin{eqnarray*}
E(L_{ik}L_{jk}D_k) &=& E\left( E(L_{ik} L_{jk}|D_k)D_k\right) \\
&=& E\left( \left[\frac{(U^{\succ j})^{-1}_{ij}}{D_k} + e_{ki} e_{kj} \right]  D_k\right) \mbox{ where $e_{ki}$ is the $i$-th entry of }\\
&=& (U^{\succ k})^{-1}_{ij} + e_{ki} e_{kj} \frac{\delta_k + \nu_k +2}{c_k} \\
\end{eqnarray*}
Similarly for $i>j$,
\[E(L_{ij} D_j) = e_{ji} \frac{\delta_j + \nu_j +2}{c_j}\]
Adding the expectations above gives us the required result, i.e.,
\[ E(\Omega) = \sum_{k<p}[(U^{\succ k})^{-1}]^0 +  e^T \mathcal{H} e \]

\subsection{Proof of Lemma \ref{lem4.0}}
\label{proof4.0}
(a) $\Rightarrow$ (b): It follows that for $r \leq k$ the power of any $\widetilde{D}_r$ in $L_{ik}$ and $L_{jk}$ can be 
$0,-1$ only. Hence the power of $\widetilde{D}_r$ in $L_{ik} L_{jk}$ can be $-2,-1,0$ only and hence the power of $
\widetilde{D}_r$ in $L_{ik} L_{jk}D_k = L_{ik} L_{jk}\widetilde{D}_1 \times \widetilde{D}_2 \times \ldots \times 
\widetilde{D}_k$ can be $-1,0,1$ only.

\smallskip

\noindent
(b) $\Rightarrow$ (a): Suppose to the contrary that Property-B does not hold. Then there exist $i > j$ with $(i,j) \notin E_\sigma$ and  $r<j$ such that the maximal negative power of $\widetilde{D}_r$ in the expansion of $L_{ij}$ is at least $-2$. Let $-k$ denote this power. Then the expansion of $L_{ij}^2D_j$ has at least one term where the power of $\widetilde{D}_r$ is $-2k + 1$. Since $-2k+1 \leq -3$, this contradicts (b).

\subsection{Proof of Lemma \ref{nocutting}} 
\label{proofnocutting}
Firstly, note that if either of $L_{ik}$ or $L_{jk}$ is functionally zero, then by assumption $L_{ik'} L_{jk'} D_{k'}$ is functionally non-zero, and we are done. Hence, without loss of generality, we consider a situation where both $L_{ik}$ and $L_{jk}$ are functionally non-zero. 
\begin{description}
\item[Case 1:] Suppose $(i,k) \in E_\sigma$ or $(j,k) \in E_\sigma$. Then $L_{ik} L_{jk} D_k$ 
is functionally dependent on atleast one of $L_{ik}$ or $L_{jk}$, whereas $L_{ik'} L_{jk'} D_{k'}$ is functionally 
independent of $L_{ik}$ and $L_{jk}$ by (\ref{powerexpan}). 
\item[Case 2:] Suppose $(i,k) \notin E_\sigma$ and $(j,k) \notin E_\sigma$. Then by (\ref{powerexpan}), each term in the expansion of $L_{ik}L_{jk}D_{k}$ is functionally dependent on $\widetilde{D}_k$. However, since $k' < k$, every term in the expansion of $L_{ik'} L_{jk'} D_{k'}$ is functionally independent of $\widetilde{D}_k$ (by (\ref{powerexpan})). 
\end{description}

\subsection{Proof of Lemma \ref{genbarimpliesab}} 
\label{proofotherone}

Suppose to the contrary that either Property-A or Property-B is violated.
\begin{description}
\item[Case 1::] Property-A does not hold.\\
Let $L_{ij}$ be the first (dependent) entry where it is violated. Thus $\exists$ $k<j$, s.t. $L_{ik}L_{jk}\frac{D_k}{D_j}$ has a square term. Now if any one of them is independent, say $L_{ik}$, then $L_{jk}$ has the square term because $L_{jk}$ can't have $L_{ik}$ in it's expansion. But that violates the assumption that $L_{ij}$ is the first term. Hence we have, $i>j>k$ such that,
\begin{eqnarray*}
(i,j) \notin E &\&\hspace{2mm} (i,k) \notin E &\& \hspace{2mm} (j,k) \notin E \\
\mbox{and     } L_{ij} \neq 0 &\&\hspace{2mm} L_{ik} \neq 0 &\& \hspace{2mm} L_{jk} \neq 0 \\
\end{eqnarray*}
Thus we have violated Generalized Bartlett Property.
\item[Case 2::] Property-B does not hold.\\
Let $L_{ij}$ be the first (dependent) entry where it is violated. Then  $\exists$ $i>j>k \geq s$, s.t. the power of $\widetilde{D}_s$ in $L_{ik}L_{jk}\frac{D_k}{D_j}$ is $\leq -2$. Now if any one (or both) of $L_{ik}$ and $L_{jk}$ is independent, we will get a contradiction since $L_{ij}$ is the first term. Thus we get $i>j>k$ s.t.,
\begin{eqnarray*}
(i,j) \notin E &\&\hspace{2mm} (i,k) \notin E &\& \hspace{2mm} (j,k) \notin E \\
\mbox{and     } L_{ij} \neq 0 &\&\hspace{2mm} L_{ik} \neq 0 &\& \hspace{2mm} L_{jk} \neq 0 \\
\end{eqnarray*}
contradicting Generalized Bartlett Property.
\end{description}

\subsection{Proof of Theorem \ref{expansion}}
\label{proofexpansion}
We prove that $\widetilde{G}$ satisfies Property-A and Property-B. In particular we will show that for every $i>j$, $L_{ij}$ satisfies Property A and Property B. Consider the following cases.
\begin{description}
\item[Case 1] 
Suppose
$ i \notin \{p_1,p_2,\ldots ,p_r\} $ and $j \in Graph\_before(i)$.

Here $(i,1),(i,2), \ldots ,(i,j) \notin \tilde{E}$, which implies that,
\[L_{i,1} = L_{i,2} = \ldots = L_{i,j} = 0 \]

\item[Case 2] 
Suppose $i \notin \{p_1,p_2,\ldots ,p_r\} $ and $j \in Graph(i), \mbox{ and } j<i$.

If $(i,j) \in \widetilde{E}$, then $L_{i,j}$ is an independent parameter. Otherwise using Case 1 we get that,
\begin{eqnarray*}
L_{i,j} &=& -\sum_{s < j}L_{i,s} L_{j,s} \frac{D_s}{D_j} = -\sum_{s < j \mbox{, } s \in Graph(j)} L_{i,s} L_{j,s} \frac{D_s}{D_j} \\
\end{eqnarray*}
Since $i>j>s$ and $Graph(i)=Graph(j)$ is a Generalized Bartlett graph, it follows from Lemma \ref{genbarimpliesab} that $L_{ij}$ has Property A and Property B. 

\item[Case 3]
Suppose $i \in \{p_1,p_2, \ldots, p_r \}$ and $j \in Graph\_before(i) \mbox{ and } j \notin \{p_1,p_2, \ldots, p_r \}$.

Since $(i,j) \notin \widetilde{E}$ again using Case 1, we get
\begin{eqnarray*}
L_{i,j} &=& -\sum_{s<j} L_{i,s}L_{j,s} \frac{D_s}{D_j} = -\sum_{s<j \mbox{ and } s \in Graph(j)} L_{i,s}L_{j,s} \frac{D_s}{D_j} \\
\end{eqnarray*}
Now suppose $Graph(j) = G_k$. Since $(i,p_{k-1}+1) \notin \widetilde{E}$, it follows that 
\[L_{i,p_{k-1}+1} = -\sum_{s < p_{k-1}+1} L_{i,s}L_{p_{k-1}+1,s} \frac{D_s}{D_{p_{k-1}+1}} = 0  \]
since by Case 1, $L_{p_{k-1}+1,s} = 0$.
\noindent
Now, we can show inductively for $p_{k-1}+2,\ldots,j$;  
\[L_{i,p_{k-1}+2}=0, \ldots, L_{i,j}=0\].

\item[Case 4]
Suppose $i \in \{p_1,p_2, \ldots, p_r \}$ and $j \in Graph(i)$.
If $(i,j) \in \widetilde{E}$ then $L_{i,j}$ is an independent parameter. Otherwise by Case 1, 
\begin{eqnarray*}
L_{i,j} &=& -\sum_{s<j} L_{i,s}L_{j,s} \frac{D_s}{D_j} = -\sum_{s<j \mbox{, } s \in Graph(j)} L_{i,s}L_{j,s} \frac{D_s}{D_j}\\
\end{eqnarray*}
Similar arguments as that in Case 2 can now be used to show that $L_{i,j}$  satisfies Property A and Property B.

\item[Case 5]
Suppose $i>j \mbox{ and } i,j \in \{p_1,p_2, \ldots ,p_r\}$.

If $(i,j) \in \widetilde{E}$ then $L_{i,j}$ is an independent parameter.If not,
\[L_{i,j} = -\sum_{s<j}L_{i,s}L_{j,s}\frac{D_s}{D_j}\]
Note that every $s<j$ belongs to $Graph\_before(i)$. By Case 3, for $s \notin \{p_1,p_2, \ldots, p_r\}$,  $L_{i,s}=0$.
Hence, 
\[ L_{i,j} = -\sum_{s<j \mbox{ and } s \in \{p_1,p_2, \ldots ,p_r\}}L_{i,s}L_{j,s}\frac{D_s}{D_j}\]
Since $i>j>s$ belongs to $\{p_1,p_2, \ldots ,p_r\}$ and $G$ is a Generalized 
 Bartlett graph it follows from Lemma \ref{genbarimpliesab} $L_{ij}$ has Property A and Property B.\end{description}

\subsection{Proof of Theorem \ref{tree}}
\label{prooftree}
We will prove Property-A and Property-B by considering the following cases. In particular we will show that for every $k>k'$, $L_{kk'}$ satisfies Property A and Property B. 
\begin{description}
\item[Case 1]
Suppose $k \in G_i^{(j)},k' \in G_{i'}^{(j')}$ for $(i,j)\neq(i',j')$ and $k>k'$.Here $k$ is not a neighbor of any of $\{1,2,\ldots,k'\}$. Thus $L_{kk'}=0$.

\item[Case 2]
Suppose $k>k'\in G_i^{(j)}$. If $(k,k') \in \widetilde{E}$, then $L_{kk'}$ is an independent parameter. Otherwise,
\begin{eqnarray*}
L_{kk'} &=& -\frac{1}{D_{k'}} \sum_{k'' < k'} L_{kk''} L_{k'k''} D_{k''} = -\frac{1}{D_{k'}} \sum_{k'' < k',\, k'' \in G_i^{(j)}} L_{kk''} L_{k'k''} D_{k''}.
\end{eqnarray*}
Since $k>k'>k''$ belong to the same Generalized Bartlett graph it follows from Lemma \ref{genbarimpliesab}, $L_{kk'}$ satisfies Property A and Property B.

\item[Case 3]
Suppose $k\in T$ and $k' \in G_k^{(j)}$. Then $L_{kk'}$ is independent.

\item[Case 4]
Suppose $k\in T$ and $k' \in G_{i'}^{(j')}$ where $k \neq i'$. Note that $(k,k') \notin \widetilde{E}$, and hence by Case 1, 
\begin{eqnarray*}
L_{kk'} &=& -\frac{1}{D_{k'}} \sum_{k'' < k'} L_{kk''} L_{k'k''} D_{k''} = -\frac{1}{D_{k'}} \sum_{k'' < k',\, k'' \in G_{i'}^{(j')}} L_{kk''} L_{k'k''} D_{k''} 
\end{eqnarray*}
If $k_1$ is the smallest labeled element of $G_{i'}^{(j')}$, $L_{kk_1}=0$. From here by induction in can be proved that $L_{kk'}=0$.
\end{description}
Cases 1 through 4 prove that for all $k>k'$, $L_{kk'}$ satisfies Property A and Property B.

\subsection{Proof of Lemma \ref{lemgrid}}

We shall prove the Generalized Bartlett property by induction on $n$. Suppose $n=2$. We will now express the dependent entries in terms of the independent entries and see whether any quadratic terms appear or not.
\begin{eqnarray*}
L_{31} &=& L_{51} = L_{61} = L_{62} = 0\\
L_{42} &=& -L_{41}L_{21} \frac{D_1}{D_2}\\
L_{43} &=& -L_{41}L_{31} \frac{D_1}{D_3}-L_{42}L_{32} \frac{D_2}{D_3} = L_{41} L_{21} L_{32} \frac{D_1 D_2}{D_3} \\
L_{53} &=& -L_{52}L_{32} \frac{D_2}{D_3}, \hspace{.5cm} L_{64} = -L_{63}L_{43} \frac{D_3}{D_4} \\
\end{eqnarray*}
Thus for $(n=2) \times 3$ grid, Property-A and Property-B is satisfied. Now suppose it is satisfied for $(n=k) \times 3$ grid, we want to prove for a $(n=k+1) \times 3$ grid.\\

\noindent
Again we will express the dependent entries in terms of the independent entries and see whether any quadratic terms appear or not. The induction hypothesis implies that Property A and Property B are satisfied for $L_{ij}$ with $i \leq 3k$. Note that,
\begin{eqnarray*}
L_{3k+1,1} &=& L_{3k+1,2} = \ldots = L_{3k+1,3k-3} =0 \\
L_{3k+1,3k-1} &=& -L_{3k+1,3k-2}L_{3k-1,3k-2} \frac{D_{3k-2}}{D_{3k-1}} \\
L_{3k+1,3k} &=& -L_{3k+1,3k-2} L_{3k,3k-2} \frac{D_{3k-2}}{D_{3k}} -L_{3k+1,3k-1} L_{3k,3k-1} \frac{D_{3k-1}}{D_{3k}}\\
&=& -L_{3k+1,3k-2} L_{3k,3k-2} \frac{D_{3k-2}}{D_{3k}} +L_{3k+1,3k-2} L_{3k-1,3k-2} L_{3k,3k-1} \frac{D_{3k-2}}{D_{3k}} \\
&& \hspace{3cm} \mbox{[$L_{3k,3k-2}$ does not contain $L_{3k+1,3k-2}$]}\\ 
L_{3k+2,1} &=& \ldots=L_{3k+2,3k-3} = L_{3k+2,3k-2}=0 \\
L_{3k+2,3k} &=& -L_{3k+2,3k-1}L_{3k,3k-1} \frac{D_{3k-1}}{D_{3k}} \\
L_{3k+3,1} &=& L_{3k+3,2}= \ldots=L_{3k+3,3k-3} = L_{3k+3,3k-2}=L_{3k+3,3k-1}=0 \\
L_{3k+3,3k+1} &=& -L_{3k+3,3k}L_{3k+1,3k} \frac{D_{3k}}{D_{3k+1}} \hspace{1cm} \mbox{[ $L_{3k+1,3k}$ does not contain $L_{3k+3,3k}$]} \\
\end{eqnarray*}
Thus a $(n=k+1) \times 3$ grid satisfies Property-A. We now proceed to check Property-B. Since $L_{3k,3k-2}$ and $L_{3k+1,3k}$ satisfy Property-B, so does $L_{3k,3k-2}\frac{D_{3k-2}}{D_{3k}}$ and $L_{3k+1,3k} \frac{D_{3k}}{D_{3k+1}}$. Hence by induction, and Theorem 4, for all $n \geq 2$, $n \times 3$ grid satisfies the Generalized Bartlett property.

\section{Comparison with Letac-Massam distributions}
\noindent
\label{completacmassam}
For a decomposable graph $G=(V,E)$, \cite{letacmassam} generalized the Wishart distribution, by defining two classes of distributions with multiple shape parameters on the cones $\mathcal{Q}_G$ and $\mathbb{P}_G$. They are called Type I and Type II Wishart distributions and have proved to be useful in high-dimensional Bayesian inference, as shown in \cite{rajamassamcarv}. Let $I_G$ denote the set of incomplete $p \times p$ matrices $X$, where $X_{ij}$ is missing iff $(i,j) \notin E$. Recall that $\mathcal{Q}_G$ is defined as 
\[\mathcal{Q}_G = \{X \in I_G| X_C \mbox{ is positive definite for all cliques } C \mbox{ in } G \}.\]
For $U \in \mathcal{Q}_G$, let $\hat{U}$ be the unique $p \times p$ matrix such that, $\hat{U}_{ij}=U_{ij}$ for $i=j$ and $(i,j) \in E$, and $\hat{U}^{-1} \in \mathbb{P}_G$. Also for a $p \times p$ matrix $X$, let $\kappa(X)$ be symmetric incomplete matrix such that $(i,j)$-th entry is missing if $(i,j) \notin E$, and $\kappa(X)_{ij}=X_{ij}$ otherwise. 

It is natural to compare and contrast our generalized $G$-Wishart distributions with the Letac-Massam distributions when $G$ is decomposable.

First, we consider the $W_{\mathbb{P}_G}$ family of Type II Wishart distributions (defined on the space $\mathbb{P}_G$) in \cite{letacmassam}. In particular, the $W_{\mathbb{P}_G}^{U,\bm{\alpha,\beta}}$ density on $\mathbb{P}_G$ is proportional to 
\begin{eqnarray*}
W_{\mathbb{P}_G}^{U,\bm{\alpha,\beta}} \left(\Omega \right) &\propto & \exp \left(-tr(\Omega U)/2 \right)  \times \frac{\prod_{C \in \mathcal{C}} |(\Omega^{-1})_C|^{\alpha(C)+(c+1)/2}}{\prod_{S \in \mathcal{S}} |(\Omega^{-1})_S|^{\nu(S)(\beta(S)+(s+1)/2)}}. 
\end{eqnarray*}
 
\noindent
Clearly, the exponential term in the above density is the same as the exponential term in the generalized 
$G$-Wishart density in (\ref{defggwshrt}). Now let us compare terms outside the exponential. For the 
generalized $G$-Wishart density in (\ref{defggwshrt}), the non-exponential term is 
\[ D_{11}^{\delta_1/2} D_{22}^{\delta_2/2} D_{33}^{\delta_3/2} D_{44}^{\delta_4/2}, \]

\noindent
where $\Omega = LDL^T$ is the modified Cholesky decomposition of $\Omega$. The corresponding term for 
$W_{\mathbb{P}_G}$ is \[\frac{\prod_{C \in \mathcal{C}} |(\Omega^{-1})_C|^{\alpha(C)+(c+1)/2}}{\prod_{S 
\in \mathcal{S}} |(\Omega^{-1})_S|^{\nu(S)(\beta(S)+(s+1)/2)}}.\]

\noindent
To contrast these two terms, we consider the case when the graph $G$ is the $4$-chain, $A_4$, given by 
$\bullet-\bullet-\bullet-\bullet$. Hence,  $C_1=\{1,2\},C_2=\{2,3\},C_3=\{3,4\}$ and $S_2=\{3\},S_3=\{4\}$. It follows that 
\begin{eqnarray*}
\frac{\prod_{i=1}^3 |(\Omega^{-1})_{C_i}|^{\alpha_i}}{\prod_{i=2}^3 |(\Omega^{-1})_{S_i}|^{\beta_i}} &=& \left(\frac{1}{D_{11}} \right)^{\alpha_1} \left(\frac{1}{D_{22}} +\frac{L_{32}^2}{D_{33}} +\frac{L_{32}^2 L_{43}^2}{D_{44}} \right)^{\alpha_1-\beta_1} \\
&\times & \left(\frac{1}{D_{22}} \right)^{\alpha_2} \left(\frac{1}{D_{33}} +\frac{L_{43}^2}{D_{44}} \right)^{\alpha_2-\beta_2} \left(\frac{1}{D_{33}D_{44}} \right)^{\alpha_3}. 
\end{eqnarray*}

\noindent
Thus even for this simple graph the non-exponential term for $W_{\mathbb{P}_G}$ is very different than 
the corresponding non-exponential term for the generalized $G$-Wishart. However, if the graph $G$ is homogeneous, then \cite{letacmassam} shows that for any clique $C$ and separator $S$,
\[|(\Omega^{-1})_C|=\prod_{i \in C} \frac{1}{D_{ii}},|(\Omega^{-1})_S|=\prod_{i \in S} \frac{1}{D_{ii}}.\] 

\noindent
In the homogeneous setting we see that the term outside the exponential is similar to that of the generalized $G$-Wishart. The family of generalized $G$-Wishart distributions introduced in the paper are therefore in general structurally different than the family of Type I and Type II Wishart distributions introduced in \cite{letacmassam}. In the special case of homogeneous graphs, the family of generalized $G$-Wishart distributions coincides with the family of Type II Wishart distributions in \cite{letacmassam}.

Next we consider the family of Type I Wishart distributions (defined on the space $\mathcal{Q}_G$), which is  refereed to as $W_{\mathcal{Q}_G}$ in \cite{letacmassam}. The family of inverse Wishart distributions induced by $W_{\mathcal{Q}_G}$ on the space $\mathbb{P}_G$ is referred to as $IW_{\mathcal{Q}_G}$. In particular, the 
$IW_{\mathcal{Q}_G}^{U,\bm{\alpha,\beta}}$ density on $\mathbb{P}_G$ is proportional to 
\begin{eqnarray*}
\exp \left(-tr(\Omega^{-1} U)/2 \right)  \times \frac{\prod_{C \in \mathcal{C}} |(\Omega^{-1})_C|^{\alpha(C)+(c+1)/2}}{\prod_{S \in \mathcal{S}} |(\Omega^{-1})_S|^{\nu(S)(\beta(S)+(s+1)/2)}}, 
\end{eqnarray*}

\noindent
where $U \in \mathcal{Q}_G$ and $\alpha(C), C \in \mathcal{C}$, $\beta(S), S \in \mathcal{S}$ are real numbers, $c=|C|,s=|S|$ and $\nu(S)$ is the multiplicity of the minimal separator $S$ which is positive and independent of the perfect order of the cliques considered (as proved by \cite{lauritzen}). Note that, as 
expected, the exponential term in the above density is $\exp \left(-tr(\Omega^{-1}U)/2 \right)$, whereas 
the exponential term in the generalized $G$-Wishart density in (\ref{defggwshrt}) is $\exp \left( - 
tr(\Omega U)/2 \right)$. Hence the difference between the two classes is fundamental.

\section{Model Selection Example} \label{modelselectionexp}
We now demonstrate through a simulation experiment that the methodology proposed in the paper is competitive with standard methods for high-dimensional graphical model selection. 

For a given dimension $p$, a ``true" sparse graph $G = (V,E)$ with $p$ vertices is chosen by taking a simple random sample (without replacement) of size $\frac{p*(p-1)}{2}*0.01$ from the total number of possible edges. We consider five different values for the number of variables $p$, ranging from $p=50$ to $p=1000$. The sample size is chosen to be $100$, $200$ or $300$. Then, a ``true" precision matrix $\Omega_0 \in \mathbb{P}_G$ is generated by taking $\Omega_0=LDL^T$, where $D$ is the identity matrix and $L_{ij}$ is a some constant (depending on $p$) if $\{i>j, \, (i,j) \in E\}$ else $L_{ij}=-\frac{1}{D_j} \sum_{k<j}L_{ik}L_{jk}D_{k}$. Then, $n$ {\it i.i.d.} samples from a $N(0,\Omega_0^{-1})$ distribution are generated. Let $S$ denote the sample covariance matrix of these $n$ samples. 

The goal now is to estimate the original graph $G$. Our approach is as follows. We first obtain a collection of ``good" models (equivalently, graphs) by using the popular penalized graphical model selection method \cite{glasso}, and then use our Bayesian approach to select the best model out of this collection. The Glasso method takes a penalty parameter $\rho$ as an input, and for a given value of $\rho$ provides a sparse estimate of the inverse covariance matrix $\Omega$. The sparsity pattern in the estimate of $\Omega$ in turn leads to an estimate of the underlying graph/model. \cite{banerjee} propose a simple and popular method for choosing the penalty parameter $\rho$ for Glasso, and thereby choosing a graph. 

Our model selection algorithm works in conjunction with Glasso, the penalized likelihood algorithm introduced in \cite{glasso}. We shall consider a grid of penalty parameters for Glasso and consider models with various levels of sparsity. Before applying the Glasso algorithm we standardize the covariance matrix $S$. In this case, it is known that $\rho=1$ produces extremely sparse models and $\rho$ close to $0$ produces extremely dense models. Our penalty parameter grid starts with $\rho=1$ and ends at $0.01$ and decreases by steps of $0.02$. For each value of $\rho$ in this grid, the Glasso algorithm is run to obtain a graph estimate with adjacency matrix $M^{\rho}$. Graphs with edge density from $0\%$ to $10\%$ are considered.

For each $M^{\rho}$ thus obtained, we use the Deviance Information Criterion (DIC) as a measure of how well the estimated  graph/model fits the data. Recall from \cite{mitsakakis2011} that $DIC = 2\bar{D}-D(\bar{\Omega})$, where $D(\Omega) = n*(tr(\Omega S)-\log(|\Omega|))$, $\bar{D}$ is the posterior expectation of $D(\Omega)$, and $\bar{\Omega}$ is the posterior expectation of $\Omega$. Ideally, for each value of $\rho$, we would like to compute the DIC for the model corresponding to $M^{\rho}$. Note however, that the graph corresponding to $M^{\rho}$ may not in general be Generalized Bartlett. Thus we generate a Generalized Bartlett cover $M^{cover,\rho}$ for $M^{\rho}$ as described in Algorithm \ref{genbarcover}. If $p$ is large, and $M^{cover, \rho}$ is quite dense, then for computational reasons, we choose a decomposable cover $M^{dc,\rho}$ using the R package ``igraph" . Once the appropriate cover has been computed, we compute the DIC score corresponding to this cover using hyperparameter values $U=cI_p$ (where $c$ is the mean of the diagonal entries of $nS$) and $\delta_j=(U_{jj}+nS_{jj})/S_{jj}$ for $1 \leq j \leq p$. This DIC score is treated as a measure of goodness of fit for the model corresponding to $M^\rho$. Finally, we choose the $M^\rho$ with the best goodness of fit score. For each value of $p$, the whole process is repeated $20$ times, and the average sensitivity and specificity is reported in Table \ref{banerjee}. For comparison purposes we also report the average sensitivity and specificity of the model obtained by using the $l_1$ penalized likelihood method proposed in \cite{banerjee}. 

In order to make sure that we are searching for models in a range whose edge density includes the true density ($1\%$) we use the starting value of $\rho=1$ (edge density almost $0\%$) and the algorithm ends at value of $\rho$ with edge density around $7\%$. Thus the true edge density of $1\%$ lies in the range. 

We compare the model selection performance of Glasso (using the approach in \cite{banerjee}) and the generalized $G$-Wishart based Bayesian approach outlined above, in Table \ref{banerjee}. Both approaches have very high specificity, with the Glasso based approach performing slightly better. On the contrary, the generalized $G$-Wishart approach shows an immense improvement in terms of sensitivity as compared to the Glasso approach. This is particularly useful in high dimensional biological applications where discovery of an important gene is much more important than exclusion of a  non-important one.

\begin{table}[h]
\begin{tabular}{|c|c|c|c|c|c|c|c|}
\hline
$p$ & $n$ & \multicolumn{2}{c|}{Specificity} & \multicolumn{2}{c|}{Sensitivity}\\ \hline
 & &  Glasso-Ban & gen. G-Wishart &   Glasso-Ban & gen. G-Wishart\\ \hline
50 & 100 &   1 & 0.9830 &  0.5833 & 1\\ \hline
100 & 100 &  1 & 0.9714 &  0 & 0.8663\\ \hline
200 & 100 &  1 & 0.9316 & 0.0007538 &	0.7781\\ \hline
500 & 200 &  0.9999 & 0.9166 &  0.0041 & 0.5570\\ \hline
1000 & 300 & 0.9899 & 0.9214 &  0.0023 & 0.2772\\ \hline
\end{tabular}
\caption{Model selection comparison of Glasso (with penalty parameter chosen by \cite{banerjee}) and generalized $G$-Wishart based Bayesian approach}
\label{banerjee}
\end{table}

\section{Application to Breast cancer data}
\label{breastcancerdata}
In this section, we use the methodology developed in this paper to analyze a dataset from a breast cancer study in \cite{changetal2005}. This study is based on $n=248$ patients, whose expression level of $24481$ genes are recorded. As in \cite{kharesangraja2013}, we focus on the reduced dataset of $p=1107$ genes closely associated with breast cancer. The objective is to obtain a sparse partial correlation graph, i.e., a sparse estimate of the inverse covariance matrix for the $1107$ genes, to identify the hub genes. As in Section \ref{modelselectionexp}, we shall choose candidate partial correlation graphs by using penalized likelihood/pseudo-likelihood methods, and then choose the best graph by computing the DIC score using the Bayesian methodology developed in this paper. The idea is to reduce our search space to a handful of graphs and then use the generalized Bartlett methodology developed in this paper for Bayesian model selection. To obtain the candidate graphs, we shall use four standard penalized algorithms: SYMLASSO (\cite{friedmanetal2010}), CONCORD (\cite{kharesangraja2013}), GLASSO (\cite{glasso}) and SPACE (\cite{pengetal2009}). For each of these algorithms, the respective penalty parameters are chosen so that the resulting partial correlation graph has 100 edges. All the four graphs thus obtained are not Generalized Bartlett, and we obtain Generalized Bartlett covers for each of them using Algorithm \ref{genbarcover}. All of these covers have at most $3$ extra edges as compared to the original graph. 

Note that each of these four partial correlation (cover) graphs represents a concentration graph model. Let $S$ denote the sample covariance matrix. Note that, as in \cite{kharesangraja2013}, each of the $p = 1107$ data columns were centered and scaled (with respect to mean absolute deviation) prior to computing $S$. For each of the four models, we choose a generalized Wishart prior with parameters $U$ and 
$\bm{\delta}$ as \[U = mean(diag(n*S))*I_p+n*S\] and \[\bm{\delta} = \mbox{mean}(diag(U)/diag(S)) {\bf 1}_p,\]

\noindent
where ${\bf 1}$ denotes the $p$-dimensional vector of all ones. Next, we run the Gibbs sampling algorithm for each of these four scenarios for $1000$ steps. The resulting Markov chains are used to compute the DIC score for each of the four partial correlation graphs (using the procedure from \cite{mitsakakis2011} outlined in Section \ref{modelselectionexp}). The DIC scores are provided in the Table \ref{4algorithms}, and show that the graph chosen using the CONCORD algorithm performs the best. This example illustrates that the methodology developed in this paper can be used in conjunction with DIC for high-dimensional graphical model selection in applied settings. 

\begin{table}[h]
\centering
\begin{tabular}{|p{2.5cm}|p{2cm}|}
\hline
Algorithm &  DIC \\ \hline
SYMLASSO &  298816.6 \\
CONCORD & 295766.6 \\
GLASSO & 299601.1 \\
SPACE  &  299302.8 \\ \hline
\end{tabular}
\caption{Comparison of 4 algorithms}
\label{4algorithms}
\end{table}

\end{document}